\documentclass[10pt,a4paper]{article}

\usepackage{graphicx}      
\usepackage{cite}
\usepackage{graphicx}
\usepackage{algorithmic}
\usepackage{algorithm}
\usepackage{amssymb}  
\usepackage{amsmath} 
\usepackage{amsthm}
\usepackage{mathrsfs}
\usepackage{xspace}
\usepackage{hyperref}
\usepackage{booktabs}
\usepackage{multicol}

\newtheoremstyle{ex}
  { }
  { }
  { }
  { }
  {\bfseries}
  { }
  { }
  {\thmname{#1}\thmnumber{ #2}:\thmnote{ #3}}
\theoremstyle{ex}

\theoremstyle{remark}

\theoremstyle{plain}
\newtheorem{theorem}{Theorem}
\newtheorem{proposition}{Proposition}
\newtheorem{assumption}{Assumption}

\usepackage{abbreviations}
\usepackage{lindsten}

\newlength\pffigheight
\newcommand\priorx{\pi}
\newcommand\x{\mathbf{x}}
\renewcommand\i{\mathbf{i}}

\newcommand\pg{PG\@\xspace}
\newcommand\pgbsi{PG-BSi\@\xspace}
\newcommand\mwpg{MwPG\@\xspace}
\renewcommand\parspace{\Theta}
\newcommand\kbsi{\psi^{\prime,\parameter}}

\newcommand\exind[2]{\substack{#1=1 \\ #1\neq#2}}

\newcommand\sampletime{dt}

\begin{document}
\title{On the use of backward simulation in particle Markov chain Monte Carlo methods}

\author{Fredrik Lindsten and Thomas B. Sch{\"o}n\\
Division of Automatic Control\\
Link{\"o}ping University\\
SE--$581 83$ Link{\"o}ping, Sweden.\\
E-mail: \{lindsten, schon\}@isy.liu.se}
\maketitle

\begin{abstract}
  Recently, Andrieu, Doucet and Holenstein \cite{AndrieuDH:2010} introduced a general framework for
  using particle filters (PFs) to construct proposal
  kernels for Markov chain Monte Carlo (MCMC) methods. This framework, termed Particle Markov chain Monte Carlo (PMCMC),
  was shown to provide powerful methods for joint Bayesian state and parameter inference in nonlinear/non-Gaussian
  state-space models. However, the mixing of the resulting MCMC kernels can be quite sensitive, both to
  the number of particles used in the underlying PF and to the number of observations in the data.
  In the discussion following \cite{AndrieuDH:2010}, Whiteley \cite{Whiteley:2010} suggested
  a modified version of one of the PMCMC samplers, namely the particle Gibbs (\pg) sampler,
  and argued that this should improve its mixing. In this paper we explore the consequences of this modification
  and show that it leads to a method which is much more robust to a low number of particles as well as a large number of observations.
  Furthermore, we discuss how the modified \pg sampler can be used as a basis for alternatives to all three PMCMC samplers
  derived in \cite{AndrieuDH:2010}.  We evaluate these methods on several challenging inference problems in a simulation study.
  One of these is the identification of an epidemiological model for predicting
  influenza epidemics, based on search engine query data.
\end{abstract}

\section{Introduction}\label{sec:intro}%
Particle Markov chain Monte Carlo (PMCMC) is an umbrella for a collection of  methods introduced in \cite{AndrieuDH:2010}.
The fundamental idea underlying these methods, is to use a sequential Monte Carlo (SMC) sampler, \ie a particle
filter (PF), to construct a proposal kernel for an MCMC sampler. The resulting methods were shown to be powerful tools
for joint Bayesian parameter and state inference in nonlinear, non-Gaussian state-space models.
However, to obtain reasonable mixing of the resulting MCMC kernels, it was reported
that a fairly high number of particles was required in the underlying SMC samplers. This might seem like a very natural observation;
in order to obtain a good proposal kernel based on a PF, we should intuitively use sufficiently many particles to obtain
high accuracy in the PF. However, this is also one of the major criticisms against PMCMC. If we need to run a PF with a high number of particles
at \emph{each iteration} of the PMCMC sampler, then a lot of computational resources will be wasted and
the method will be very computationally intense (or ``computationally brutal'' as Flury and
Shephard put it in the discussion following \cite{AndrieuDH:2010}).
It is the purpose of this work to show that this need not be the case. We will discuss alternatives to each of the
three PMCMC methods derived in \cite{AndrieuDH:2010}, that will function properly even when using very few particles.
The basic idea underlying these modified PMCMC samplers was originally proposed by Whiteley \cite{Whiteley:2010}.

To formulate the problem that we are concerned with,
consider a general, discrete-time state-space model with state-space $\setX$, parameterised by $\parameter \in \parspace$,
\begin{subequations}
  \label{eq:intro_ssm}
  \begin{align}
    x_{t+1} &\sim f_\parameter(x_{t+1} \mid x_t), \\
    y_t &\sim g_\parameter(y_t \mid x_t).
  \end{align}
\end{subequations}
The initial state has density $\priorx_\parameter(x_1)$.
We take a Bayesian viewpoint and model the parameter as a random variable with prior density $p(\parameter)$.
Given a sequence of observations $y_{1:\T} \triangleq \crange{y_1}{y_\T}$, we wish to estimate the parameter $\parameter$ as well as
the system state $x_{1:\T}$. That is, we seek the posterior density $p(\parameter, x_{1:\T} \mid y_{1:\T})$.

During the last two decades or so, Monte Carlo methods for state and parameter inference in this type of nonlinear, non-Gaussian state-space models
have appeared at an increasing rate and with increasingly better performance. State inference via SMC is thoroughly treated by for instance
\cite{Gustafsson:2010a,DoucetJ:2011,CappeMR:2005}.
For the problem of parameter inference, the two overview papers
\cite{AndrieuDST:2004,KantasDSM:2009} provide a good introduction to both frequentistic and Bayesian methods.
Some recent results on Monte Carlo approaches to maximum likelihood parameter inference in state-space models can be found in
\cite{SchonWN:2011,OlssonDCM:2008,Gopaluni:2008,PoyiadjisDS:2011}.
Existing methods based on PMCMC will be discussed in the next section,
where we also provide a preview of the material presented in the present work.

\section{An overview of PMCMC methods}\label{sec:preview}
In \cite{AndrieuDH:2010}, three PMCMC methods were introduced to address the inference problem mentioned above. These methods
are referred to as particle Gibbs (\pg), particle marginal Metropolis-Hastings (PMMH) and particle independent Metropolis-Hastings (PIMH).

Let us start by considering the \pg sampler, which is an MCMC
 method targeting the joint density $p(\parameter, x_{1:\T} \mid y_{1:\T})$.
In an ``idealised'' Gibbs sampler (see \eg \cite{Liu:2001} for an introduction to Gibbs sampling),
we would target this density by the following two-step sweep,
\begin{itemize}
  \item Draw $\theta^\star \mid x_{1:\T} \sim p(\theta \mid x_{1:\T}, y_{1:\T})$.
  \item Draw $x_{1:\T}^\star \mid \theta^\star \sim p_{\theta^\star}(x_{1:\T} \mid y_{1:\T})$.
\end{itemize}
The first step of this procedure can, for some problems be carried out exactly
(basically if we use conjugate priors for the model under study). In the \pg sampler, this is assumed to be the case.
However, the second step, \ie to sample from the joint smoothing density
$p_{\theta}(x_{1:\T} \mid y_{1:\T})$, is in most cases very difficult. In the \pg sampler,
this is addressed by instead sampling a particle trajectory $x_{1:\T}^\star$ using a PF. More precisely, we
run a PF targeting the joint smoothing density. We then sample one of the particles at the final time $\T$, according to their
importance weights, and trace the ancestral lineage of this particle to obtain the trajectory $x_{1:\T}^\star$.
However, as we shall see in \Sec{pgbsi_example} this can lead to very poor mixing when there is degeneracy in the PF (see \eg \cite{GodsillDW:2004,DoucetJ:2011}
for a discussion on PF degeneracy). This will inevitably be the case when $\T$ is large and/or $\Np$ is small.

One way to remedy this problem is to append a backward simulator to the \pg sampler,
leading to a method that we denote particle Gibbs with backward simulation (\pgbsi).
The idea of using backward simulation in the \pg context was mentioned by Whiteley \cite{Whiteley:2010} in the
discussion following \cite{AndrieuDH:2010}. To make this paper self contained, we will present a
derivation of the \pgbsi sampler in \Sec{pgbsi}.
The reason for why \pgbsi can avoid the poor mixing of the original \pg sampler will be discussed in \Sec{pgbsi_discussion}.
Furthermore, in \Sec{pgbsi_example} we shall see that the \pgbsi sampler can operate properly even with very few particles and vastly outperform \pg
in such cases.

Now, as mentioned above, to apply the \pg and the \pgbsi samplers we need to sample from the conditional parameter density
$p(\parameter \mid x_{1:\T}, y_{1:\T})$. This is not always possible, since conjugate priors can be both
 undesirable and unavailable (\eg if there is a nonlinear parameter dependence in the model).
For such cases, the PMMH of \cite{AndrieuDH:2010} can be used. Similarly to \pg, this method targets the joint posterior
density $p(\parameter, x_{1:\T} \mid y_{1:\T})$, but parameter values are drawn from some proposal kernel
$q(\parameter \mid \parameter^\prime)$ in a Metropolis-Hastings (MH) fashion. Hence, the method does not require
the density $p(\theta \mid x_{1:\T}, y_{1:\T})$ to be available for sampling.

In this work, we propose an alternative method called Metropolis within particle Gibbs (\mwpg), derived in \Sec{mwpg}.
The \mwpg sampler is based on the \pgbsi sampler,
but instead of sampling exactly from $p(\theta \mid x_{1:\T}, y_{1:\T})$
we use an MH step to sample the parameters. This might at first seem superfluous and appear as
a poor imitation of the PMMH sampler. However, as we shall see in \Sec{mwpg_example}, the \mwpg sampler can
vastly outperform PMMH, even when the two methods use the exact same proposal density to sample the parameters.
The reason is that PMMH requires accurate estimates of the likelihood $p_\parameter(y_{1:\T})$
to get a high acceptance probability \cite{AndrieuDH:2010}. This implies that also this method requires a high number of particles in the PF.
On the contrary, as mentioned above, the \pgbsi sampler has the ability to function properly even for a very low number of particles.
Since they share a common base, this desirable property of the \pgbsi sampler will ``rub off'' on the \mwpg sampler.

Finally, the PIMH introduced in \cite{AndrieuDH:2010} is as a method targeting the joint smoothing
density $p_\parameter(x_{1:\T} \mid y_{1:\T})$ when the parameter $\parameter$ is assumed to be known. Hence, it is a competitor
to existing particle smoothing algorithms, see \eg \cite{GodsillDW:2004,BriersDM:2010,DoucGMO:2011,FearnheadWT:2010}.
The PIMH can be seen as a collapsed version of the PMMH, and will suffer
from the same problem when we decrease the number of particles. To mitigate this, we propose to instead use a collapsed version of the
\pgbsi sampler. The details are given in \Sec{smoothing} and the proposed method is compared with the PIMH and a state of the art particle smoother
\cite{DoucGMO:2011} in \Sec{smoothing_example}.

Since the publication of the seminal paper \cite{AndrieuDH:2010}, several contributions have been made which relate to the present work.
We have published a preliminary version of the \pgbsi sampler in \cite{LindstenS:2012}. We have also applied this to the
problem of Wiener system identification in \cite{LindstenSJ:2012}.
In \cite{WhiteleyAD:2010}, a version of \pg with backward simulation is used for inference in switching linear dynamical systems.
In this work, it was observed that the backward sampling step could indeed drastically improve the performance of the \pg sampler.
In \cite{WhiteleyAD:2011} a similar approach is used for multiple change-point models.

Backward simulation has also been used in the context of PIMH and PMMH \cite{OlssonR:2011}.
However, it should be noted that their approach is fundamentally different from the \mwpg sampler suggested in this work.
In \cite{OlssonR:2011}, the aim is to ``improve the quality'' of the sampled state trajectories. The acceptance probabilities
in the MH samplers suggested in \cite{OlssonR:2011} are identical to those used in the original PMMH and PIMH samplers. Hence,
they will suffer from the same problems when using few particles.

\section{Particle filter and backward simulator}\label{sec:pfandbsi}%
Before we go in to the details of the \pgbsi method in \Sec{pgbsi}, we review the PF and the backward simulator. Throughout this
section we assume that the parameter $\parameter$ is fixed and discuss how to approximately sample from the joint smoothing density
$p_\parameter(x_{1:\T} \mid y_{1:\T})$.

\subsection{Auxiliary particle filter}
The auxiliary particle filter (APF) \cite{PittS:1999} can be used to approximate (and approximately sample from) the joint
smoothing distribution. The reason for why we choose an APF, rather than a ``standard'' PF, is that it is more general
and that the notion of auxiliary variables to determine ancestor indices is well suited for the derivation of the \pgbsi sampler in \Sec{pgbsi}.

We shall assume that the reader is familiar with the APF and make the presentation quite brief, mainly just to introduce the
required notation. For readers not familiar with the PF or the APF, we refer to \cite{Gustafsson:2010a,DoucetJ:2011,PittS:1999}.
Let $\{x_{1:t-1}^m, w_{t-1}^m\}_{m=1}^\Np$ be a \wps targeting $p_\parameter(x_{1:t-1} \mid y_{1:t-1})$. That is, the particle system
defines an empirical distribution,
\begin{align*}
  \widehat p_\parameter^\Np (dx_{1:t-1} \mid y_{1:t-1}) \triangleq \sum_{m=1}^\Np \frac{w_{t-1}^m}{\normsum{l} w_{t-1}^l} \delta_{x_{1:t-1}^m }(dx_{1:t-1}),
\end{align*}
which approximates the target distribution. In the APF, we propagate this sample to time $t$ by sampling from a proposal kernel,
\begin{align*}
  M_t^\parameter(i_t, x_t)
  \triangleq \frac{w_{t-1}^{i_t} \nu_{t-1}^{i_t} }{\normsum{l} w_{t-1}^l \nu_{t-1}^l}  R_t^\parameter(x_t \mid x_{t-1}^{i_t}).
\end{align*}
Here, the variable $i_t$ is the index to an ``ancestor particle'' $x_{t-1}^{i_t}$ and $R_t^\parameter$ is a proposal kernel which proposes a new particle
at time $t$ given this ancestor. The factors $\nu_{t-1}^{i_t} = \nu_{t-1}^\parameter(x_{t-1}^{i_t}, y_{t})$, known as adjustment multiplier weights, are used in the APF to increase the probability of
sampling ancestors that better can describe the current observation.
If these adjustment weights are identically equal to 1, we recover the standard PF.

Once we have generated $\Np$ new particles (and ancestor indices) from the kernel $M_t^\parameter$, the particles are assigned importance weights
according to $ w_{t}^m = W_{t}^\parameter(x_t^m, x_{t-1}^{i_t^m})$
for $m = \range{1}{\Np}$, where the weight function is given by,
\begin{align}
  \label{eq:pf_W}
  W_{t}^\parameter(x_t, x_{t-1}) \triangleq \frac{g_\parameter(y_t \mid x_t) f_\parameter(x_t \mid x_{t-1}) }{ \nu_{t-1}^\parameter(x_{t-1}, y_t) R_t^\parameter(x_t \mid x_{t-1})}.
\end{align}
This results in a new \wps $\{x_{1:t}^m, w_{t}^m\}_{m=1}^\Np$, targeting the joint smoothing density at time $t$.

The method is initialised by sampling from a proposal density $x_1^m \sim R_1^\parameter(x_1)$ and weighting the samples according to
$w_1^m = W_1^\parameter(x_1^m)$ where the weight function is given by,
$W_{1}^\parameter(x_1) \triangleq g_\parameter(y_1 \mid x_1) \priorx_\parameter(x_1) / R_1^\parameter(x_1)$.

After a complete run of the APF we can sample a trajectory from the empirical joint smoothing density by choosing
particle $x_{1:\T}^m$ with probability proportional to $w_\T^m$. Note that this particle trajectory consists of the
ancestral lineage of $x_\T^m$. Hence, if this lineage is defined by the indices $b_{1:\T}$ (with $b_\T = m$) we
obtain a sample
\begin{align}
  \label{eq:pf_sampletraj}
  x_{1:\T}^\star = x_{1:\T}^m = \crange{x_1^{b_1}}{x_\T^{b_\T}}.
\end{align}

For later reference, we note that the proposal kernel used in the APF is not allowed to be completely arbitrary.
In principle, we need to make sure that the support of the proposal covers the support of the target. Hence, we make
the following assumption.

\begin{assumption}\label{assumption:pf}%
  Let $\mathcal{S} = \{\parameter \in \parspace : p(\parameter) > 0\}$. Then, for any $\parameter \in \mathcal{S}$ and any $t \in \crange{1}{\T}$,
  $\mathcal{S}_t^\parameter \subset \mathcal{Q}_t^\parameter$ where we have defined,
  \begin{align*}
    \mathcal{S}_t^\parameter &= \{x_{1:t} \in \setX^t : p_\parameter(x_{1:t} \mid y_{1:t}) > 0 \}, \\
    \mathcal{Q}_t^\parameter &= \{x_{1:t} \in \setX^t : R_t^\parameter(x_t \mid x_{t-1}) p_\parameter(x_{1:t-1} \mid y_{1:t-1}) > 0 \}.
  \end{align*}
\end{assumption}

\subsection{Backward simulation}
In the original \pg sampler by \cite{AndrieuDH:2010}, sample trajectories are generated as in \eqref{eq:pf_sampletraj}.
However, due to the degeneracy of the APF, the resulting Gibbs sampler can suffer from poor mixing. We will illustrate this in \Sec{pgbsi_example}.
As previously mentioned, we will focus on an alternative way to generate an approximate sample trajectory from the joint smoothing density, known as backward
simulation. It was introduced as a particle smoother in \cite{DoucetGW:2000,GodsillDW:2004}.
In \cite{Whiteley:2010}, it was suggested as a possible way to increase the mixing rate of a \pg kernel.

Consider a factorisation of the joint smoothing density as, 
\begin{align*}
   p_\parameter(x_{1:\T} \mid y_{1:\T}) = p_\parameter(x_{\T} \mid y_{1:\T}) \prod_{t=1}^{\T-1} p_\parameter(x_t \mid x_{t+1}, y_{1:t}).
\end{align*}
Here, the \emph{backward kernel} density
can be written as,
\begin{align*}
  p_\parameter(x_t \mid x_{t+1}, y_{1:t}) = \frac{f_\parameter(x_{t+1} \mid x_t) p_\parameter(x_{t} \mid y_{1:t})}{p_\parameter(x_{t+1} \mid y_{1:t})}.
\end{align*}
We note that the backward kernel depends on the filtering density $p_\parameter(x_t \mid y_{1:t})$.
The key enabler of the backward simulator
is that this density (in many cases) can be readily approximated
by an (A)PF, without suffering from degeneracy. By using the APF to approximate the filtering density,
we obtain an approximation of the backward kernel,
\begin{align*}
  \widehat p_\parameter^\Np(dx_t \mid x_{t+1}, y_{1:t})
  \triangleq \sum_{m=1}^\Np \frac{w_t^m f_\parameter(x_{t+1} \mid x_t^m)}{\normsum{l} w_t^l f_\parameter(x_{t+1} \mid x_t^l)}
  \delta_{x_t^m}(dx_t).
\end{align*}
Based on this approximation, we may sample a particle trajectory backward in time, approximately distributed
according to the joint smoothing density. The procedure is given in \Alg{bsi}. The algorithm generates a set of indices $j_{1:\T}$ defining
a backward trajectory according to
\begin{align}
  \label{eq:bsi_sampletraj}
  x_{1:\T}^\star = \crange{x_t^{j_1}}{x_\T^{j_\T}}.
\end{align}
Note that, since we only need to generate a single backward trajectory,
the computational cost of the backward simulator is $\Ordo(\Np)$, \ie of the same order as the APF.
Hence, the computational complexity of the \pg sampler is more or less unaffected by the introduction of the backward simulation step.
Still, it is possible to reduce the complexity of the backward simulator even further by making use of rejection sampling
\cite{DoucGMO:2011}.

\begin{algorithm}[ptb]
  \caption{Backward simulator}
  \label{alg:bsi}
  1. \textbf{Initialise:} Set $j_\T = m$ with probability $w_{\T}^{m} / \normsum{l} w_{\T}^l $.
  
  \noindent
  2. \textbf{For $t = \bwdrange{\T-1}{1}$ do:}
  \begin{enumerate}
  \item[(a)] Given $x_{t+1}^{j_{t+1}}$, compute the smoothing weights,
    \begin{align*}
      w_{t\mid \T}^m = \frac{ w_t^{m} f_\parameter(x_{t+1}^{j_{t+1}} \mid x_t^m) }
      { \normsum{l} w_t^{l} f_\parameter(x_{t+1}^{j_{t+1}} \mid x_t^l) }, \qquad m = \range{1}{\Np}.
    \end{align*}
  \item[(b)] Set $j_t = m$ with probability $w_{t\mid \T}^m$.
    \end{enumerate}
\end{algorithm}

\section{Particle Gibbs with backward simulation}\label{sec:pgbsi}%
In this section we will derive the \pgbsi sampler and discuss its properties.
This forms the basis for all the methods proposed in this work.

\subsection{The \pgbsi sampler}
Recall the idealised Gibbs sampler, briefly outlined in \Sec{preview}. The two steps of the method which are performed at
each iteration are,
\begin{subequations}
  \begin{align}
    \theta^\star \mid x_{1:\T} &\sim p(\theta \mid x_{1:\T}, y_{1:\T}), \\
    \label{eq:pg_idealgibbs_b}
    x_{1:\T}^\star \mid \theta^\star &\sim p_{\theta^\star}(x_{1:\T} \mid y_{1:\T}).
  \end{align}
\end{subequations}
As previously mentioned, the idea behind the \pg sampler by \cite{AndrieuDH:2010} is,
instead of sampling according to \eqref{eq:pg_idealgibbs_b}, to run an APF
to generate a sample trajectory $x_{1:\T}^\star$ as in \eqref{eq:pf_sampletraj}.
In the \pgbsi sampler, we instead generate a sample trajectory
by backward simulation, as in \eqref{eq:bsi_sampletraj}.
However, in either case, to simply replace step \eqref{eq:pg_idealgibbs_b} of the idealised Gibbs sampler
with such an approximate sampling procedure will not result in a valid approach.
In \cite{AndrieuDH:2010} they solve this by replacing the APF with a so called conditional APF (CAPF),
and show that the resulting \pg sampler is a valid MCMC method. In this section,
we make a similar derivation for the \pgbsi sampler.

The key in deriving the \pgbsi sampler is to consider a Markov chain in which the state consists of
all the random variables generated by the particle filter and the backward simulator. Hence, the state of the Markov chain
will not only be the particle trajectory $\range{x_1^{j_1}}{x_\T^{j_\T}}$, but the complete set of particles, particle indices
and backward trajectory indices. For this cause, let us start by determining the density of
all the random variables generated by the APF. Let $\x_{t} = \{\range{x_t^1}{x_t^\Np}\}$ and
$\i_{t} = \{\range{i_t^1}{i_t^\Np}\}$ be the particles and ancestor indices, respectively, generated by the APF at time $t$.
Note that $i_t^m$ is (the index of) the ancestor of particle $x_t^m$. For a fixed $\parameter$ the APF then samples from,
\begin{align*}
  \psi^\parameter(\x_{1:\T}, \i_{2:\T}) = \prod_{l = 1}^\Np R_1^\parameter(x_1^l) \prod_{t = 2}^\T
  \prod_{m = 1}^\Np
  M_t^\parameter(i_t^m, x_t^m).
\end{align*}
Note that the APF is simply a method which generates a single sample from $\psi^\parameter$,
on the extended space $\setX^{\Np\T} \times \crange{1}{\Np}^{\Np(\T-1)}$. It can be beneficial to adopt this view on the APF in order to
understand the inner workings of the PMCMC samplers.

By completing this with a backward simulation according to \Alg{bsi}, we generate the indices $\range{j_1}{j_\T}$ defining
a backward trajectory. The joint density of $\{\x_{1:\T}, \i_{2:\T}, j_{1:\T}\}$ is then given by,
\begin{align}
  \label{eq:pg_psiplus}
  \kbsi(\x_{1:\T}, \i_{2:\T}, j_{1:\T}) = \psi^\parameter(\x_{1:\T}, \i_{2:\T})
  \times \frac{w_\T^{j_\T}}{\normsum{l} w_\T^l} \prod_{t = 2}^\T \frac{w_{t-1}^{j_{t-1}} f_\parameter(x_{t}^{j_t} \mid x_{t-1}^{j_{t-1}})}{\normsum{l}  w_{t-1}^l f_\parameter(x_t^{j_t} \mid x_{t-1}^l)}.
\end{align}
Thus, the state trajectory $x_{1:\T}^{j_{1:\T}} \triangleq \crange{x_1^{j_1}}{x_\T^{j_\T}}$
will be distributed (jointly with the random indices $j_{1:\T}$) according to a marginal of \eqref{eq:pg_psiplus},
$\kbsi( x_{1:\T}^{j_{1:\T}}, j_{1:\T})$, which in general is not equal to the joint smoothing density.
In other words, the extracted trajectory $x_{1:\T}^{j_{1:\T}}$ is not distributed according to $p_\parameter( x_{1:\T} \mid y_{1:\T})$
and to simply use this trajectory in step \eqref{eq:pg_idealgibbs_b} of the idealised Gibbs sampler would not be a valid approach.

The key to circumventing this problem is to introduce a new, artificial target density, with two properties:
\begin{enumerate}
\item The artificial target density should admit the joint smoothing density as one of its marginals.
\item The artificial target density should be ``similar in nature'' to $\psi^\parameter$, to enable the application of an APF in the sampling procedure.
\end{enumerate}
Here, we define this artificial target density according to,
\begin{align}
  \nonumber
  \phi(\parameter, \x_{1:\T}, \i_{2:\T}, j_{1:\T}) &\triangleq \frac{p(\parameter, x_{1:\T}^{j_{1:\T}} \mid y_{1:\T})}{\Np^\T} \\
  \label{eq:phidef}
  &\hspace{-1cm}\times \frac{\psi^\parameter(\x_{1:\T},
    \i_{2:\T})}{R_1^\parameter(x_1^{j_1}) \prod_{t=2}^\T
    M_t^\parameter( i_t^{j_t}, x_t^{j_t} )}
  \prod_{t=2}^\T \frac{w_{t-1}^{i_t^{j_t}} f_\parameter(x_t^{j_t}
    \mid x_{t-1}^{i_t^{j_t}})}{ \normsum{l} w_{t-1}^l
    f_\parameter(x_t^{j_t} \mid x_{t-1}^l)}.
\end{align}
This construction is due to \cite{OlssonR:2011}, who apply backward simulation in the context of PMMH and PIMH.
The artificial target density \eqref{eq:phidef} differs from the one used in \cite{AndrieuDH:2010} in the last factor.

One of the important properties of the density $\phi$ is formalised in the following proposition.

\begin{proposition}\label{prop:pgbsi_marginal}%
  The marginal density of $\{\parameter, x_{1:\T}^{j_{1:\T}}, j_{1:\T}\}$ under $\phi$ is given by,
  \begin{align*}
    \phi(\parameter, x_{1:\T}^{j_{1:\T}}, j_{1:\T}) = \frac{p(\parameter, x_{1:\T}^{j_{1:\T}} \mid y_{1:\T})}{\Np^\T}.
  \end{align*}
  Furthermore, with $\{\parameter, \x_{1:\T}, \i_{2:\T}, j_{1:\T}\}$ being a sample distributed according to $\phi$, the density of
  $\{\parameter, x_{1:\T}^{j_{1:\T}} \}$ is given by $p(\parameter, x_{1:\T}^{j_{1:\T}} \mid y_{1:\T})$.
\end{proposition}
\begin{proof}
  See \cite[Theorem~1]{OlssonR:2011}.
\end{proof}

This proposition has the following important implication;
if we can construct an MCMC method with stationary distribution $\phi$,
then the sub-chain given by the variables $\{\parameter, x_{1:\T}^{j_{1:\T}} \}$ will have stationary distribution
$p(\parameter, x_{1:\T} \mid y_{1:\T})$.

The construction of such an MCMC sampler is enabled by the second important
property of our artificial density, namely that it is constructed in such a way that we are able to sample from
the conditional $\phi(\x_{1:\T}^{-j_{1:\T}}, \i_{2:\T} \mid \parameter, x_{1:\T}^{j_{1:\T}}, j_{1:\T})$.
Here we have introduced the notation $\x_{1:\T}^{-j_{1:\T}} = \crange{\x_1^{-j_1}}{\x_\T^{-j_\T}}$ and
$\x_t^{-j_t} = \{\range{x_t^1}{x_t^{(j_t)-1}},\,\range{x_t^{(j_t)+1}}{x_t^\Np}\}$.
This can be done by running a so called conditional APF (CAPF), introduced in \cite{AndrieuDH:2010}\footnote{
  The CAPF used in this work is slightly different from the conditional SMC introduced in \cite{AndrieuDH:2010}.
  The difference lies in step 2(b) -- where we sample new ancestor indices, they simply set $i_t^{j_t} = j_{t+1}$. The
  reason for this difference is that the indices $j_{1:\T}$ have different interpretations in the two methods. Here
  they correspond to backward trajectory indices, whereas in the original \pg sampler they correspond to an ancestral path.}.
  This method is given in \Alg{capf} and its relationship to $\phi$ is formalised in \Prop{pgbsi_capf}.

\begin{algorithm}[ptb]
  \caption{CAPF -- conditional auxiliary particle filter (conditioned on $\{x_{1:\T}^\prime, j_{1:\T}\}$)}
  \label{alg:capf}
  1. \textbf{Initialise:}
  \begin{enumerate}
  \item[(a)] Sample $x_1^m \sim R_1^\parameter(x_1)$ for $m \neq j_1$ and set $x_1^{j_1} = x_1^\prime$.
  \item[(b)] Set $w_1^m = W_1^\parameter(x_1^m)$ for $m = \range{1}{\Np}$.
  \end{enumerate}
  \noindent
  2. \textbf{For $t = \range{2}{\T}$ do:}
  \begin{enumerate}
  \item[(a)] Sample $\{i_t^m, x_{t}^m\} \sim  M_t^\parameter(i_t, x_t)$ for $m \neq j_t$ and set $x_t^{j_t} = x_t^\prime$.
  \item[(b)] Sample $i_t^{j_t}$ according to,
    \begin{align*}
      \probab(i_t^{j_t} = m) = \frac{w_{t-1}^{m} f_\parameter(x_t^{j_t} \mid x_{t-1}^{m})}{\normsum{l} w_{t-1}^l f_\parameter(x_t^{j_t} \mid x_{t-1}^l)}.
    \end{align*}
  \item[(c)] Set $w_t^m = W_{t}^\parameter(x_t^m, x_{t-1}^{i_t^m})$ for $m = \range{1}{\Np}$.
  \end{enumerate}
\end{algorithm}

\begin{proposition}\label{prop:pgbsi_capf}%
  The conditional $\phi(\x_{1:\T}^{-j_{1:\T}}, \i_{2:\T} \mid \parameter, x_{1:\T}^{j_{1:\T}}, j_{1:\T})$ under $\phi$ is given
  by,
  \begin{align*}
    &\phi(\x_{1:\T}^{-j_{1:\T}}, \i_{2:\T} \mid \parameter, x_{1:\T}^{j_{1:\T}}, j_{1:\T}) \\
    &= \frac{\psi^\parameter(\x_{1:\T}, \i_{2:\T})}{R_1^\parameter(x_1^{j_1}) \prod_{t=2}^\T M_t^\parameter( i_t^{j_t}, x_t^{j_t} ) }
    \prod_{t=2}^\T \frac{w_{t-1}^{i_t^{j_t}} f_\parameter(x_t^{j_t} \mid x_{t-1}^{i_t^{j_t}})}{ \normsum{l} w_{t-1}^l f_\parameter(x_t^{j_t}\mid x_{t-1}^l)}.
  \end{align*}
  Consequently, the CAPF given in \Alg{capf}, conditioned on $\{ \parameter, x_{1:\T}^{j_{1:\T}}, j_{1:\T} \}$,
  generates a sample from this conditional distribution.
\end{proposition}
\begin{proof}
 See \App{proofs}.
\end{proof}
Note that the conditional density considered in \Prop{pgbsi_capf} only depends on known quantities, as opposed to the density $\phi$
defined in \eqref{eq:phidef} which depends explicitly on the joint posterior $p(\parameter, x_{1:\T} \mid y_{1:\T})$ (which in general
is not known).

As a final component of the \pgbsi sampler we need a way to generate a sample state trajectory. As previously
pointed out, we aim to do this by running a backward simulator. The fact that the backward simulator indeed generates a sample from
one of the conditionals of $\phi$ is assessed in the following proposition.

\begin{proposition}\label{prop:pgbsi_bsi}
  The conditional $\phi(j_{1:\T} \mid \parameter, \x_{1:\T}, \i_{2:\T})$ under $\phi$ is given by,
  \begin{align*}
    \phi(j_{1:\T} \mid \parameter, \x_{1:\T}, \i_{2:\T})
    = \frac{w_\T^{j_\T}}{\normsum{l} w_\T^l} \prod_{t = 2}^\T \frac{w_{t-1}^{j_{t-1}} f_\parameter(x_{t}^{j_t} \mid x_{t-1}^{j_{t-1}})}
    {\normsum{l}  w_{t-1}^l f_\parameter(x_t^{j_t} \mid x_{t-1}^l)}.
  \end{align*}
  Consequently, the backward simulator given in \Alg{bsi}, conditioned on $\{ \parameter, \x_{1:\T}, \i_{2:\T} \}$,
  generates a sample from this conditional distribution.
\end{proposition}
\begin{proof}
 See \App{proofs}.
\end{proof}
We are now ready to present the \pgbsi sampler, given in \Alg{pg_bsi}. Based on the propositions and discussion above, the three steps of
the method [2(a)--2(c)] can be interpreted as:
\begin{enumerate}
\item[(a)] Draw $\parameter^\star \sim \phi(\parameter \mid x_{1:\T}^{j_{1:\T}}, j_{1:\T})$,
\item[(b)] Draw $\x_{1:\T}^{\star,-j_{1:\T}}, \i_{2:\T}^{\star} \sim
  \phi(\x_{1:\T}^{-j_{1:\T}}, \i_{2:\T} \mid \parameter^\star, x_{1:\T}^{j_{1:\T}}, j_{1:\T})$,
\item[(c)] Draw $j_{1:\T}^\star \sim
    \phi(j_{1:\T} \mid \parameter^\star, \x_{1:\T}^{\star,-j_{1:\T}}, \i_{2:\T}^{\star},  x_{1:\T}^{j_{1:\T}})$.
\end{enumerate}
Hence, each step of \Alg{pg_bsi} is a standard Gibbs update, targeting the distribution $\phi$, and will thus leave $\phi$
invariant\footnote{In step (a) we marginalise $\phi$ over some of the variables before conditioning. This is known in the literature
as collapsed Gibbs (see \eg \cite[Sec. 6.7]{Liu:2001}) and leaves the target distribution invariant.}. It follows that $\phi$ is a stationary
distribution of the \pgbsi sampler.

It is worth to note that the sampling scheme outlined above is not a ``complete'' multistage Gibbs sweep, since it does not loop over all
the variables of the model. More precisely, we do not sample new locations for the particles $x_{1:\T}^{j_{1:\T}}$ in any of the
steps above (the same is true also for the original \pg sampler). This might raise some doubts about ergodicity of the
resulting Markov chain. However, as established in \Thm{pgbsi} below, the \pgbsi sampler is indeed ergodic
(as long as the idealised Gibbs sampler is), meaning that it is a valid MCMC method.
Intuitively, the reason for this is that the collection of variables that is ``left out'' is chosen randomly at each iteration.

\begin{assumption}\label{assumption:ergodicity}
  The idealised Gibbs sampler, defined by drawing alternately from $p(\parameter \mid x_{1:\T}, y_{1:\T})$ and
  $p_\parameter(x_{1:\T} \mid y_{1:\T})$ is $p$-irreducible and aperiodic (and hence converges for $p$-almost all
  starting points).
\end{assumption}

\begin{theorem}\label{thm:pgbsi}
  For any number of particles $\Np \geq 2$,
  \begin{enumerate}
  \item the \pgbsi sampler of \Alg{pg_bsi} defines a transition kernel on the space
    $\parspace \times \setX^{\Np\T} \times \crange{1}{\Np}^{\Np(\T-1)+\T}$
    with invariant density $\phi(\parameter, \x_{1:\T}, \i_{2:\T}, j_{1:\T})$. 
  \item Additionally, if \Assumption{pf} and \Assumption{ergodicity} hold, then the \pgbsi sampler generates a sequence
    $\{\parameter(r), x_{1:\T}(r)\}$ whose marginal distribution $\mathcal{L}^\Np(\{\parameter(r), x_{1:\T}(r)\} \in \cdot )$ satisfy,
    \begin{align*}
      \| \mathcal{L}^\Np(\{\parameter(r), x_{1:\T}(r)\} \in \cdot ) - p(\cdot \mid y_{1:\T}) \|_{\text{TV}} \goesto 0
    \end{align*}
    as $r \goesto \infty$ for $p$-almost all starting points ($\|\cdot\|_{\text{TV}}$ being the total variation norm).
  \end{enumerate}
\end{theorem}
\begin{proof}
  See \App{proofs}.
\end{proof}

It should be noted that \cite{Whiteley:2010,WhiteleyAD:2010} motivate the use of backward simulation in the \pg sampler
in a slightly different way. They simply note that the backward simulator is a Markov kernel with invariant distribution
given by the original artificial target proposed in \cite{AndrieuDH:2010}. Hence, they view the backward simulator as
a way to resample the ancestor indices for the extracted trajectory $x_{1:\T}^{j_{1:\T}}$.

\begin{algorithm}[ptb]
  \caption{\pgbsi~ -- particle Gibbs w. backward simulation}
  \label{alg:pg_bsi}
  1. \textbf{Initialise:} Set $\parameter(0)$, $x_{1:\T}(0)$ and $j_{1:\T}(0)$ arbitrarily.

  \noindent
  2. \textbf{For $r \geq 1$, iterate:}
  \begin{enumerate}
  \item[(a)] Sample $\parameter(r) \sim p(\parameter \mid x_{1:\T}(r-1), y_{1:\T})$.
  \item[(b)] Run a CAPF targeting $p_{\parameter(r)}(x_{1:\T} \mid y_{1:\T})$, conditioned on $\{x_{1:\T}(r-1), j_{1:\T}(r-1)\}$.
  \item[(c)] Run a backward simulator to generate $j_{1:\T}(r)$. Set $x_{1:\T}(r)$ to the corresponding particle trajectory.
  \end{enumerate}
\end{algorithm}


\subsection{Numerical illustration}\label{sec:pgbsi_example}
We have argued that the \pgbsi sampler in \Alg{pg_bsi} should be preferable over the standard
\pg sampler in the sense that it produces a Gibbs sampler with better mixing properties.
Since the \pgbsi sampler is the basis for all the methods proposed in this paper, we pause
for a while and consider a numerical illustration of this improved mixing.
Hence, consider the following nonlinear state-space model,
\begin{subequations}
  \label{eq:example_nl1d}
  \begin{align}
    x_{t+1} &= \beta_1 x_t + \beta_2 \frac{x_t}{1+x_t^2} + \beta_3 \cos(1.2t) + v_t, \\
    y_t &= 0.05 |x_t|^\alpha + e_t,
  \end{align}
\end{subequations}
where $x_1 \sim \N(0,5)$, $v_t \sim \N(0, \sigma_v^2)$ and $e_t \sim \N(0, \sigma_e^2)$; here $\N(0, \sigma^2)$
denotes a zero-mean Gaussian with variance $\sigma^2$.
For the time being, we assume that $\{\beta_1, \beta_2, \beta_3, \alpha\} = \{0.5, 25, 8, 2\}$ are
fixed and consider the problem of estimating the noise variances (see \Sec{mwpg_example2} for an experiment in which all the parameters are identified).
Hence, we set $\theta = \{\sigma_v^2, \sigma_e^2\}$. The same model was used in \cite{AndrieuDH:2010} to illustrate the \pg and the PMMH samplers.

We generate a set of observations $y_{1:500}$ according to \eqref{eq:example_nl1d} with $\sigma_v^2 = 10$ and $\sigma_e^2 = 1$.
The parameter priors are modelled as inverse gamma distributed with hyper-parameters $a = b = 0.01$.
We then employ the \pg sampler of \cite{AndrieuDH:2010} and the \pgbsi sampler of \Alg{pg_bsi} to estimate the posterior parameter distribution.
Since the theoretical validity of the \pgbsi and the \pg samplers can be assessed for any number of particles $\Np \geq 2$
(see \Thm{pgbsi}), 
it is interesting to see the practical implications of using very few particles.
Hence, we run the methods four times on the same data with $\Np = 5,\,20,\,1000$ and $5000$
($\Np = 5000$ as was used in \cite{AndrieuDH:2010}).
The parameters are initialised at $\theta(0) = \begin{pmatrix} 10 & 10 \end{pmatrix}^\+$
in all experiments.

\begin{figure*}[tpb]
  \centering
  \includegraphics[width = 0.49\linewidth]{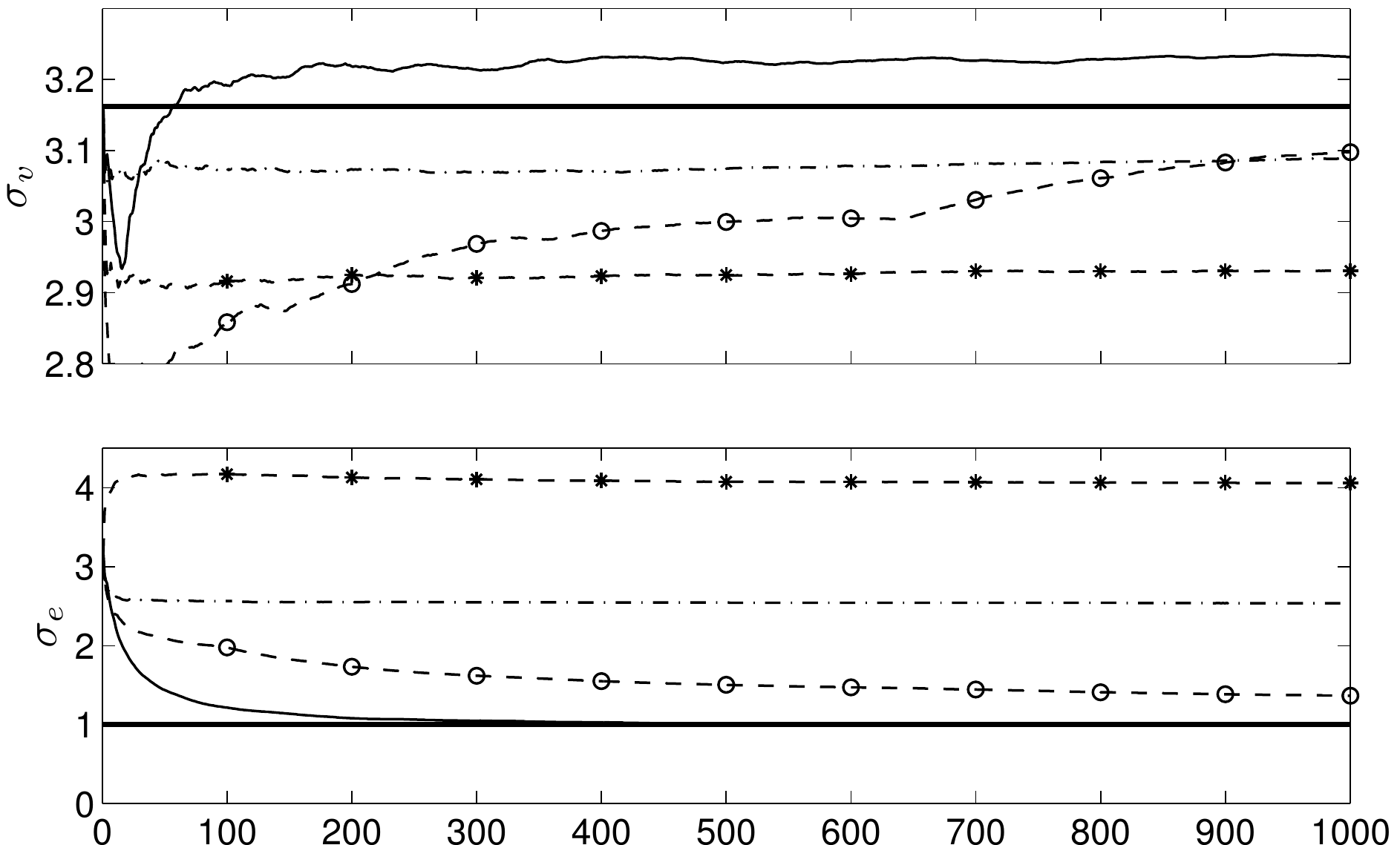}
  \includegraphics[width = 0.49\linewidth]{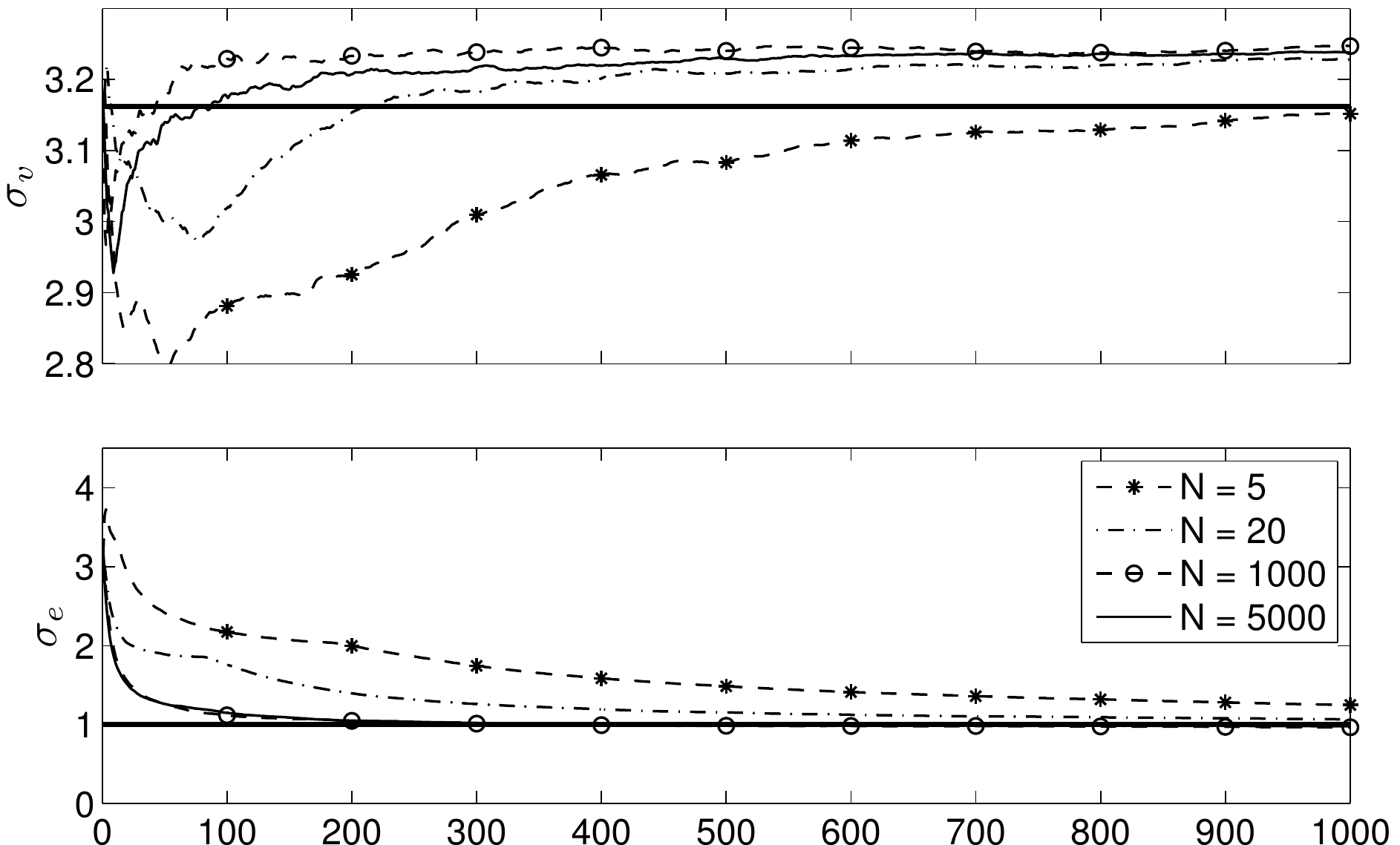}
  \caption{Running means for the estimated standard deviations $\sigma_v$ and $\sigma_e$ for the \pg sampler (left) and \pgbsi sampler (right)
    for a maximum of 1000 MCMC iterations. The four curves correspond to different number of particles. The ``true'' values are shown as thick solid lines.
Note that the estimate of $\sigma_v$ does not converge to this ``true'' value, but this is not surprising since we consider just one data realisation. Hence, the posterior mean may very well differ from the ``true'' value.}
  \label{fig:pgtest}
\end{figure*}

In \Fig{pgtest} we show the estimated posterior means of the standard deviations $\sigma_v$ and $\sigma_e$
(\ie the square roots of the means of $\parameter(1:r)$) \vs the number of MCMC iterations $r$. The left column shows the results for the \pg sampler.
When we use $\Np = 5000$ particles, there is a rapid convergence
towards the true values, indicating that the method mixes well and quickly finds a region of high posterior probability.
However, as we decrease the number of particles, the convergence is much slower. Even for $\Np = 1000$, the method struggles
and for $\Np = 20$ and $\Np = 5$ it does not seem to converge at all (in a reasonable amount of time).
The key observation is that the \pgbsi sampler (right column),
on the other hand, seems to be more or less unaffected by the large decrease in the number of particles.
In fact, the \pgbsi using just $\Np = 5$ performs equally well as the \pg sampler using $\Np = 1000$ particles.

\subsection{Why the big difference?}\label{sec:pgbsi_discussion}%
To see why the \pgbsi sampler is so much more insensitive to a decreased number of particles,
consider a toy example where we have generated $\T = 50$ observations from a 1st order linear system.
\Fig{example_tree} (left) shows the particle tree generated by the
CAPF at iteration $r$ of a \pg sampler, using $\Np = 30$ particles. Due to the degeneracy of the PF,
there is only one distinct particle trajectory for $t \leq 32$. In the \pg sampler,
we draw $x_{1:\T}(r)$ by sampling among the particles at time $\T$ and tracing the ancestral path of this
particle. This trajectory is illustrated as a thick black line in the figure. At the next iteration of the \pg sampler
we run a CAPF, conditioned on this particle trajectory. This results in the tree shown in \Fig{example_tree} (right).
Due to degeneracy, we once again obtain only a single distinct particle trajectory for $t \leq 29$.
Now, since we condition on the trajectory $x_{1:\T}(r)$, the particle tree generated at iteration
$r+1$ \emph{must} contain $x_{1:\T}(r)$. Hence, all particle trajectories available at iteration $r+1$ are identical
to $x_{1:29}(r)$ up to time $29$. Consequently, when we sample $x_{1:\T}(r+1)$ at iteration $r+1$, this trajectory will to a large
extent be identical to $x_{1:\T}(r)$. This results in a poor exploration of the state space, which in turn means that the Gibbs kernel will mix slowly.

\begin{figure}[ptb]
  \centering
  \includegraphics[height = \pffigheight]{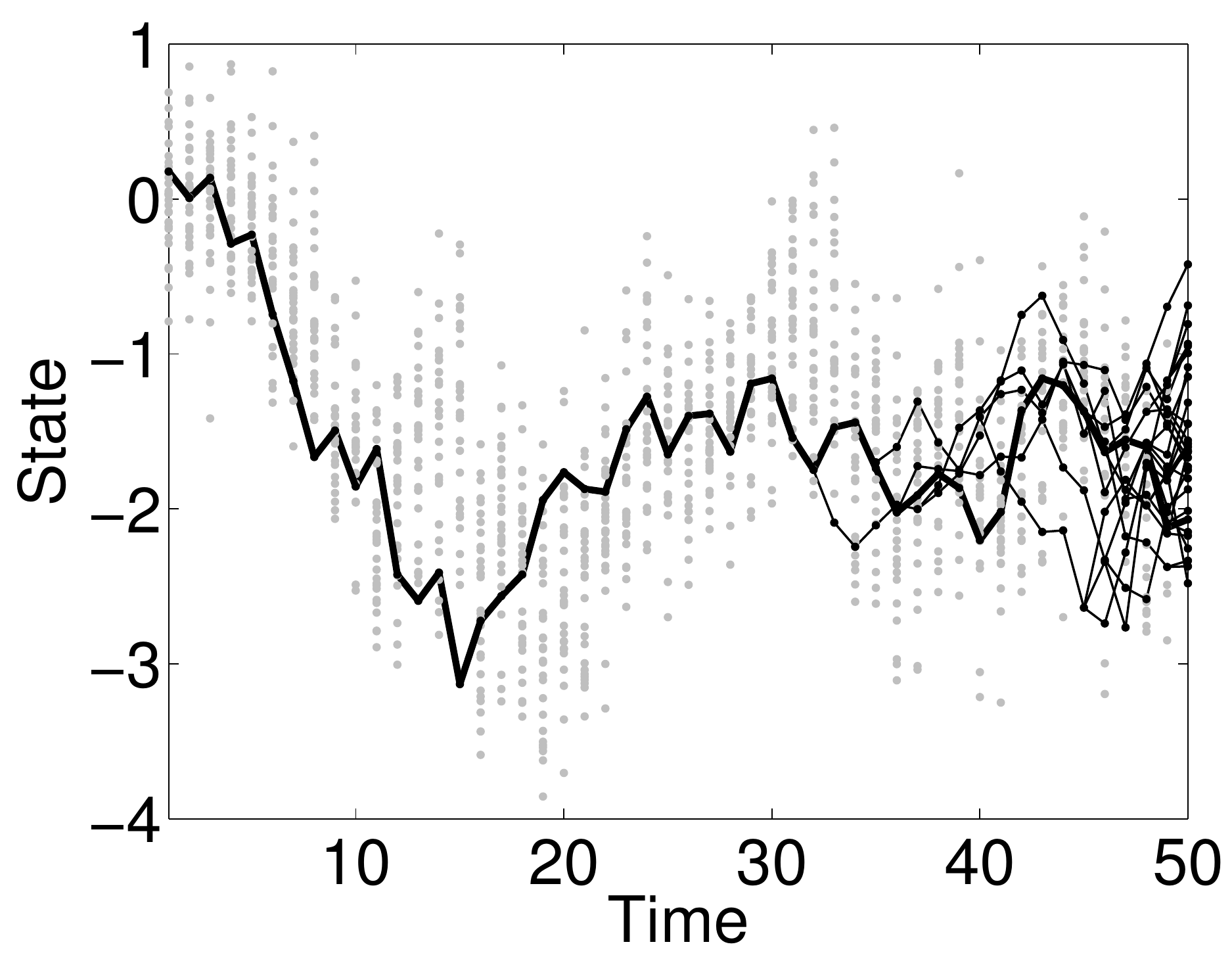}
  \includegraphics[height = \pffigheight]{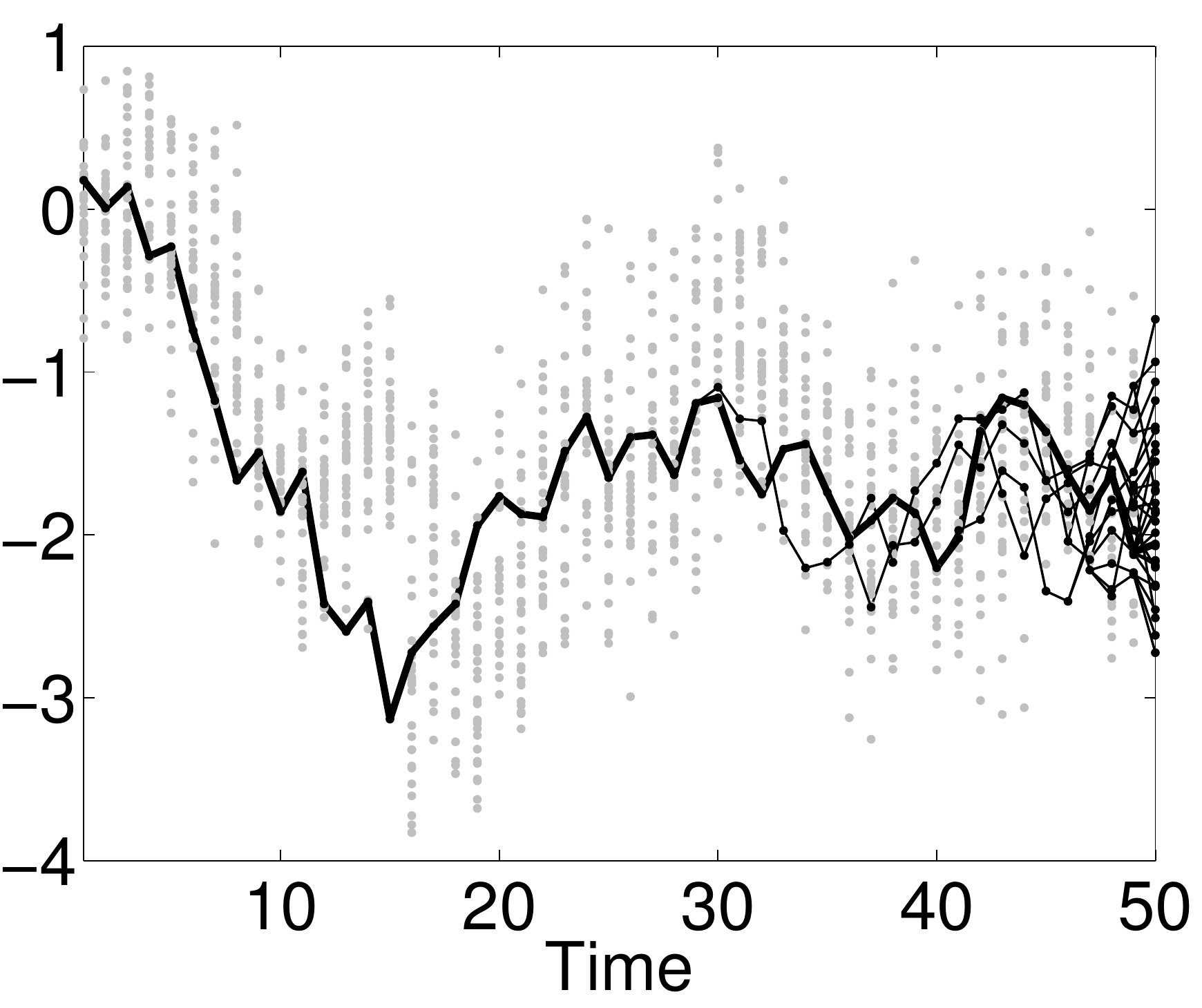}
  \caption{Particle tree generated by the CAPF at iterations $r$ (left) and $r+1$ (right) of the \pg sampler. The dots show the particle positions, the thin black lines show the ancestral dependence of the particles and the thick black lines show the sampled trajectories $x_{1:\T}(r)$ and $x_{1:\T}(r+1)$, respectively. Note that, due to degeneracy of the filter, the particles shown as grey dots are not reachable by tracing any of the ancestral lineages from time $\T$ and back.}
  \label{fig:example_tree}
\end{figure}

Based on this argument we also note that the mixing will be particularly poor when the degeneracy of the CAPF is severe. This
will be the case if the length of the data record $\T$ is large and/or if the number of particles $\Np$ is small. This
is consistent with the results reported in the previous section.

The reason for why \pgbsi can circumvent the poor mixing of the \pg kernel
is that the backward simulator explores all possible particle combinations when generating a particle trajectory,
and is thus not constrained to one of the ancestral paths. Consider the particles generated by the CAPF at iteration
$r$ of a \pgbsi sampler, shown in \Fig{example_particles_bsi} (left). A backward trajectory $x_{1:\T}(r)$, shown as a thick black line, is generated by
a backward simulator. At iteration $r+1$ of the \pgbsi sampler we will run a CAPF conditioned on $x_{1:\T}(r)$, generating the
particles shown in \Fig{example_particles_bsi} (right).
When we again apply a backward simulator to this collection of particles we will, with high probability, sample a trajectory $x_{1:\T}(r+1)$
which is entirely different from $x_{1:\T}(r)$. This leads to a much better exploration of the state space and thus a faster mixing Gibbs kernel.

\begin{figure}[ptb]
  \centering
  \includegraphics[height = \pffigheight]{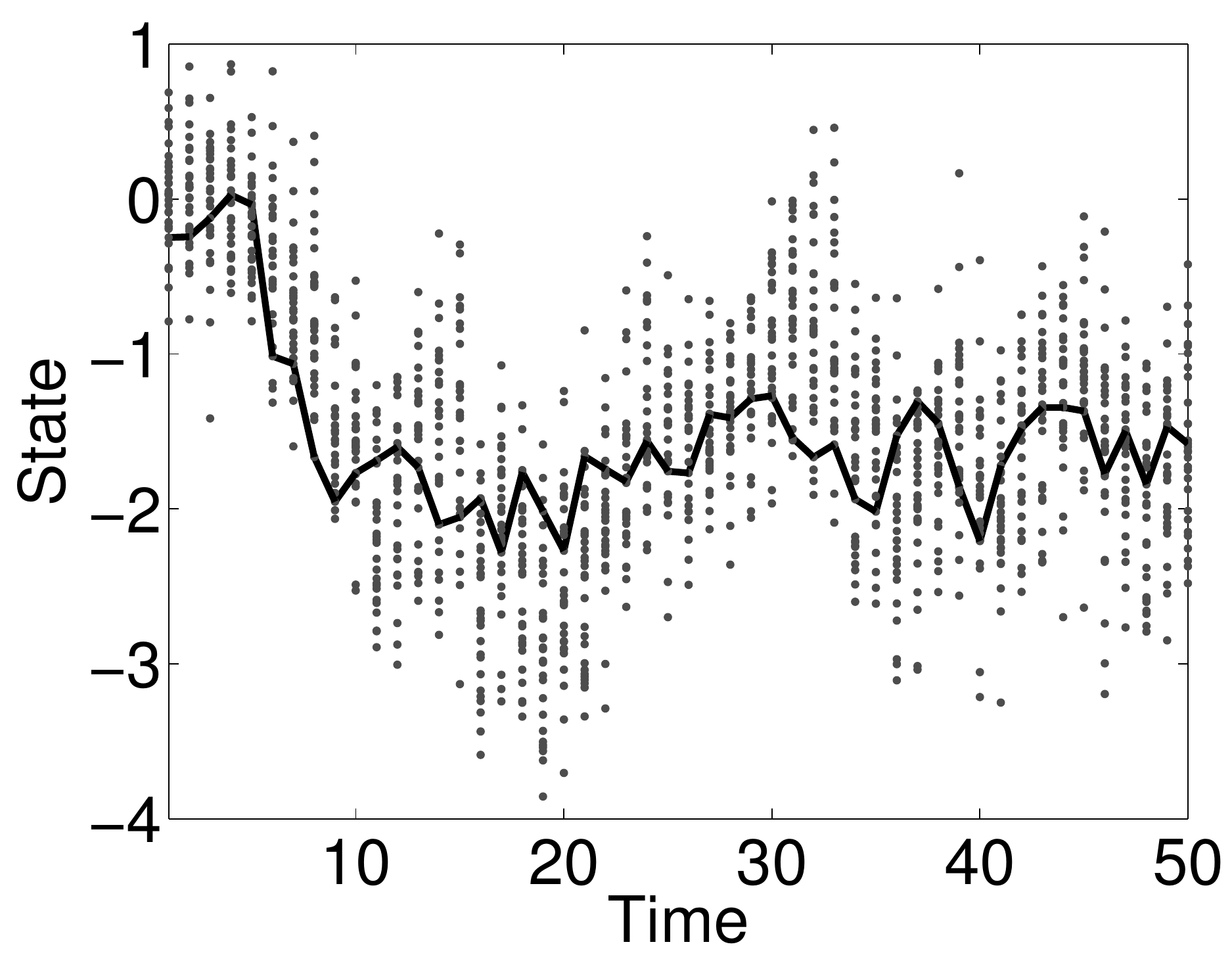}
  \includegraphics[height = \pffigheight]{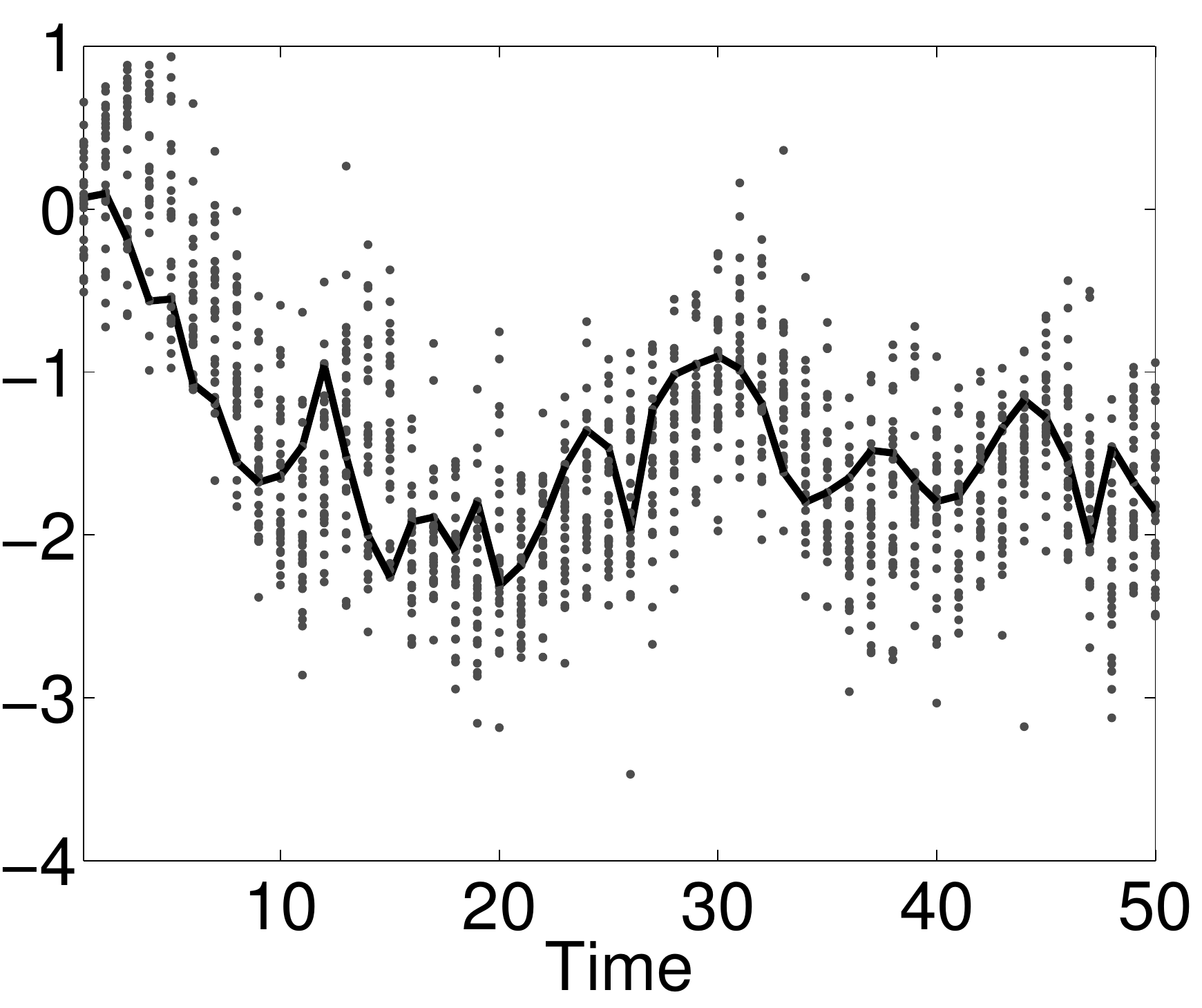}
  \caption{Particles generated by the CAPF at iterations $r$ (left) and $r+1$ (right) of the \pgbsi sampler. The dots show the particle positions and the thick black lines show the sampled backward trajectories $x_{1:\T}(r)$ and $x_{1:\T}(r+1)$, respectively. Note that all the particles at all time steps are reachable by the backward simulator.}
  \label{fig:example_particles_bsi}
\end{figure}

\begin{figure*}[tpb]
  \centering
  \includegraphics[width = 0.49\linewidth]{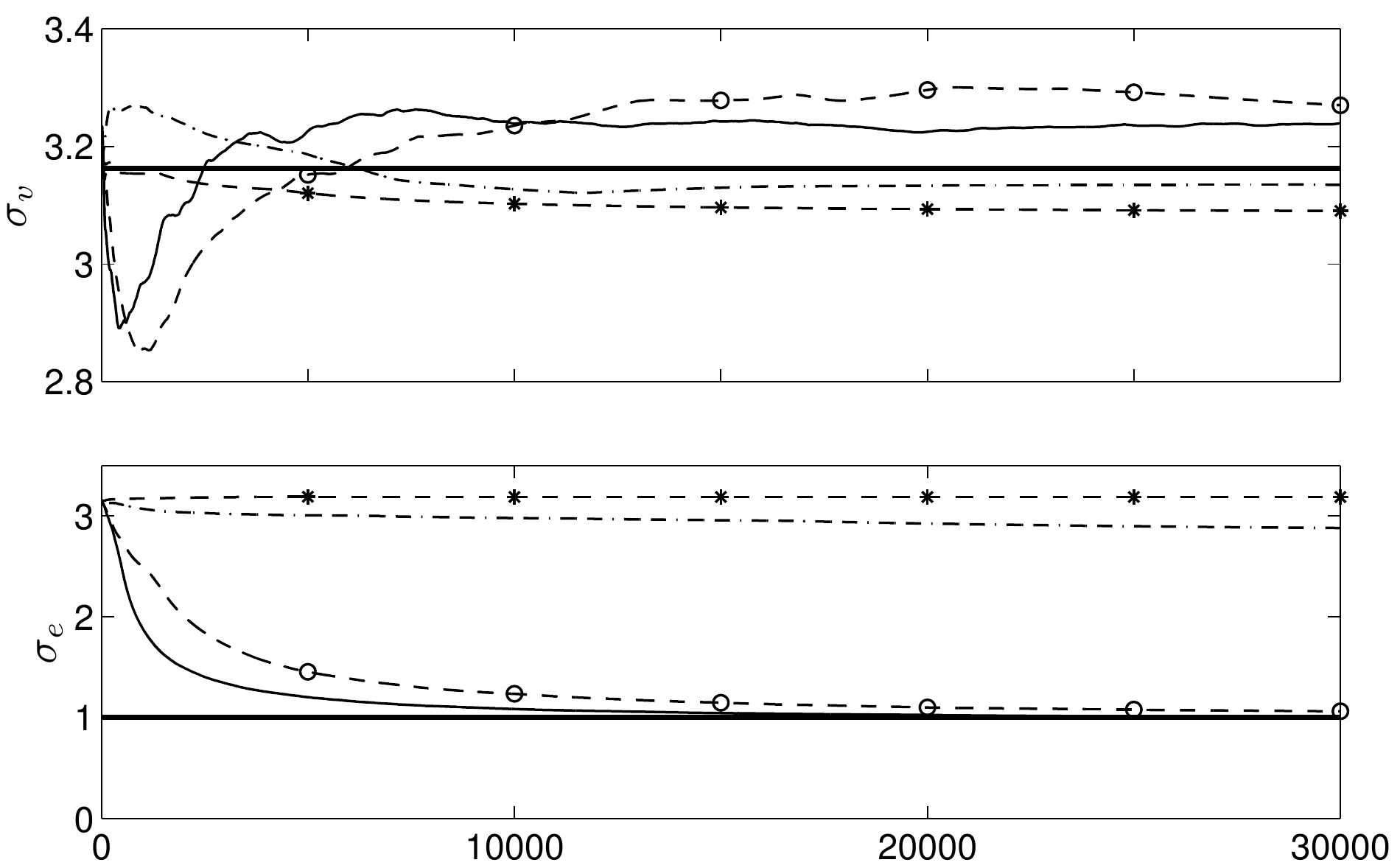}
  \includegraphics[width = 0.49\linewidth]{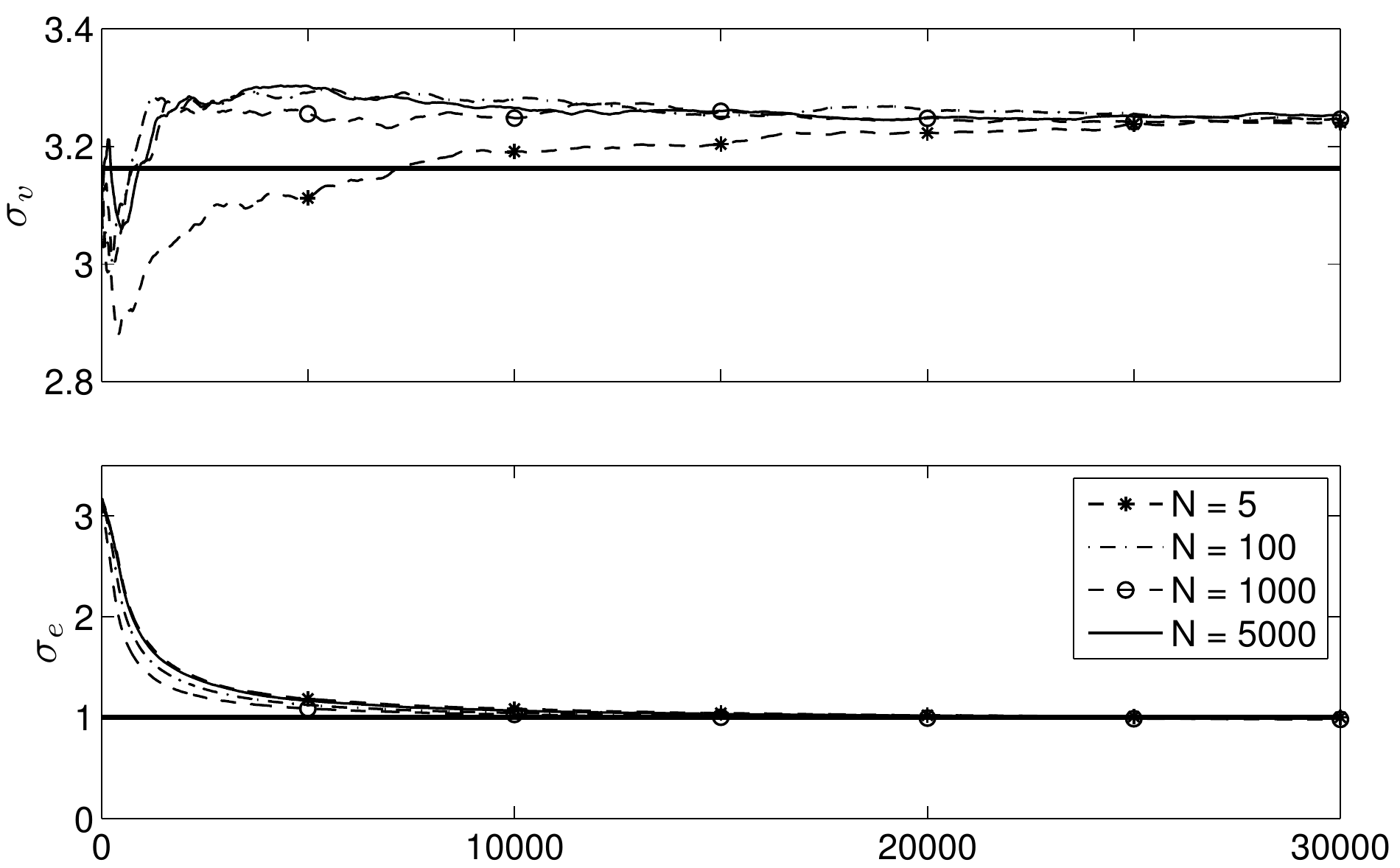}
  \caption{Running means for the estimated standard deviations $\sigma_v$ and $\sigma_e$ for the PMMH sampler (left) and MwPG sampler (right)
    for a maximum of 30000 MCMC iterations. The four curves correspond to different number of particles. The ``true'' values are shown as thick solid lines.
Note that the estimate of $\sigma_v$ does not converge to this ``true'' value, but this is not surprising since we consider just one data realisation. Hence, the posterior mean may very well differ from the ``true'' value.}
  \label{fig:mwpg_runningmeans}
\end{figure*}

\section{Metropolis within particle Gibbs}\label{sec:mwpg}
\subsection{The \mwpg sampler}
Quite often, it is not possible to sample exactly from the conditional parameter density $p(\parameter \mid x_{1:\T}, y_{1:\T})$.
If this is the case, we can replace step 2(a) in \Alg{pg_bsi} with an MH update. This results in what we refer to as a Metropolis within
particle Gibbs (\mwpg) sampler. The name stems from the commonly used term \emph{Metropolis-within-Gibbs}, which refers
to a hybrid MH/Gibbs sampler (see \cite{GilksBT:1995,Tierney:1994}). Hence, let us assume that $x_{1:\T}^\prime$ is a fixed trajectory and
consider the problem of sampling from $p(\parameter \mid x_{1:\T}^\prime, y_{1:\T})$ using MH.
We thus choose some proposal density
$q(\parameter \mid \parameter^\prime, x_{1:\T}^\prime)$, which may depend on the previous parameter $\parameter^\prime$ and the state trajectory $x_{1:\T}^\prime$.
The proposal can also depend on the fixed measurement sequence $y_{1:\T}$, but we shall not make that dependence explicit.
We then generate a sample $\parameter^{\prime\prime} \sim q(\parameter \mid \parameter^\prime, x_{1:\T}^\prime)$ and accept this with probability
\begin{align}
  \label{eq:mwpg_rhodef}
  \rho(\parameter^{\prime\prime}, \parameter^\prime, x_{1:\T}^\prime) &= 1 \wedge
  \frac{p(\parameter^{\prime\prime} \mid x_{1:\T}^\prime, y_{1:\T}) q(\parameter^\prime \mid \parameter^{\prime\prime}, x_{1:\T}^\prime) }
       {p(\parameter^\prime \mid x_{1:\T}^\prime, y_{1:\T})q(\parameter^{\prime\prime} \mid \parameter^\prime, x_{1:\T}^\prime) },
\end{align}
which is a standard MH update (see \eg \cite{Liu:2001}). The acceptance probability
can be computed directly from the quantities defining the model \eqref{eq:intro_ssm} since,
\begin{align*}
  p(\parameter \mid x_{1:\T}, y_{1:\T}) \propto \prod_{t=1}^\T g_\parameter(y_t \mid x_t)\prod_{t=1}^{\T-1} f_\parameter(x_{t+1} \mid x_t)
  \priorx_\parameter(x_1) p(\parameter).
\end{align*}

By plugging this MH step into the \pgbsi sampler we obtain the \mwpg method, presented in \Alg{mwpg}. Due to the fact that the MH step will leave the target
density invariant, the stationary distribution of the \mwpg sampler is the same as for the \pgbsi sampler. \mwpg can be seen as an alternative to
the PMMH sampler by \cite{AndrieuDH:2010}. Note that, as opposed to case of PMMH, the acceptance probability \eqref{eq:mwpg_rhodef} does not depend explicitly
on the likelihood $p(y_{1:\T} \mid \parameter)$.

\begin{algorithm}[ptb]
  \caption{\mwpg~ -- Metropolis within particle Gibbs}
  \label{alg:mwpg}
  1. \textbf{Initialise:} Set $\parameter(0)$, $x_{1:\T}(0)$ and $j_{1:\T}(0)$ arbitrarily.

  \noindent
  2. \textbf{For $r \geq 1$, iterate:}
  \begin{enumerate}
  \item[(a)] MH step for sampling a parameter:
    \begin{itemize}
    \item Sample $\parameter^{\prime\prime} \sim q(\parameter \mid \parameter(r-1), x_{1:\T}(r-1))$.
    \item Compute $\rho = \rho(\parameter^{\prime\prime}, \parameter(r-1), x_{1:\T}(r-1))$ according to \eqref{eq:mwpg_rhodef}.
    \item With probability $\rho$, set $\parameter(r) = \parameter^{\prime\prime}$, otherwise set $\parameter(r) = \parameter(r-1)$.
    \end{itemize}
  \item[(b)] Run a CAPF targeting $p_{\parameter(r)}(x_{1:\T} \mid y_{1:\T})$, conditioned on $\{x_{1:\T}(r-1), j_{1:\T}(r-1)\}$.
  \item[(c)] Run a backward simulator to generate $j_{1:\T}(r)$. Set $x_{1:\T}(r)$ to the corresponding particle trajectory.
  \end{enumerate}
\end{algorithm}

\subsection{Numerical illustration -- estimating variances}\label{sec:mwpg_example}%
Let us return to the simulation example studied in \Sec{pgbsi_example}. We use the same
batch of data $y_{1:500}$, generated from model \eqref{eq:example_nl1d} with $\sigma_v^2 = 10$ and $\sigma_e^2 = 1$.
As before, we wish to estimate the noise variances and set $\parameter = \{\sigma_v^2, \sigma_e^2\}$,
but now we apply the PMMH sampler of \cite{AndrieuDH:2010} and the \mwpg sampler in \Alg{mwpg}.
As proposal kernel for the parameter, we use a Gaussian random walk with standard deviation 0.15 for $\sigma_v^2$
and 0.08 for $\sigma_e^2$ (both methods use the same proposal).

Running means for the estimated standard deviations, for both methods, are shown in \Fig{mwpg_runningmeans}.
The difference between \mwpg and PMMH is similar to that between \pgbsi and \pg, especially when we consider the
effect of using few particles.
Compared to the \pgbsi sampler operating on the same data (see \Fig{pgtest} and note the different scaling of the axes)
the \mwpg sampler is slower to converge. However, this is expected, since the latter does not make use of the conjugacy of the priors.

In \Fig{mwpg_scatter} we show scatter plots for the Markov chains after 50000 iterations, using a burnin of 10000 iterations.
From these plots, it is clear that the estimated posterior parameter distribution for the \mwpg sampler is rather insensitive
to the number of particles used, and similar to what we get from PMMH using $\Np = 5000$. The same effect can be seen by
analysing the autocorrelation functions, shown in \Fig{mwpg_acf}.
Finally, the average acceptance probabilities for the two methods and for different number of particles are given in \Tab{mwpg_accept}.

It is worth to note that, in the limit $\Np \goesto \infty$, the PMMH will ``converge to'' an idealised, marginal MH sampler. The
\mwpg, on the other hand, ``converges to'' an idealised Metropolis within Gibbs sampler. Hence, for a very large number of particles,
we expect that PMMH should outperform \mwpg. Furthermore, and more importantly, whether or not \mwpg is preferable over PMMH should be
very much problem dependent. If there is a strong dependence between the states and the parameters, we expect that the \mwpg sampler should suffer
from poor mixing, since it samples the parameters conditioned on the states, and vice versa. In such cases, the PMMH might perform better, since
it more closely resembles a marginal MH sampler.

\begin{multicols}{2}
  \begin{figure}[H]
    \centering
    \includegraphics[height = 4.3cm]{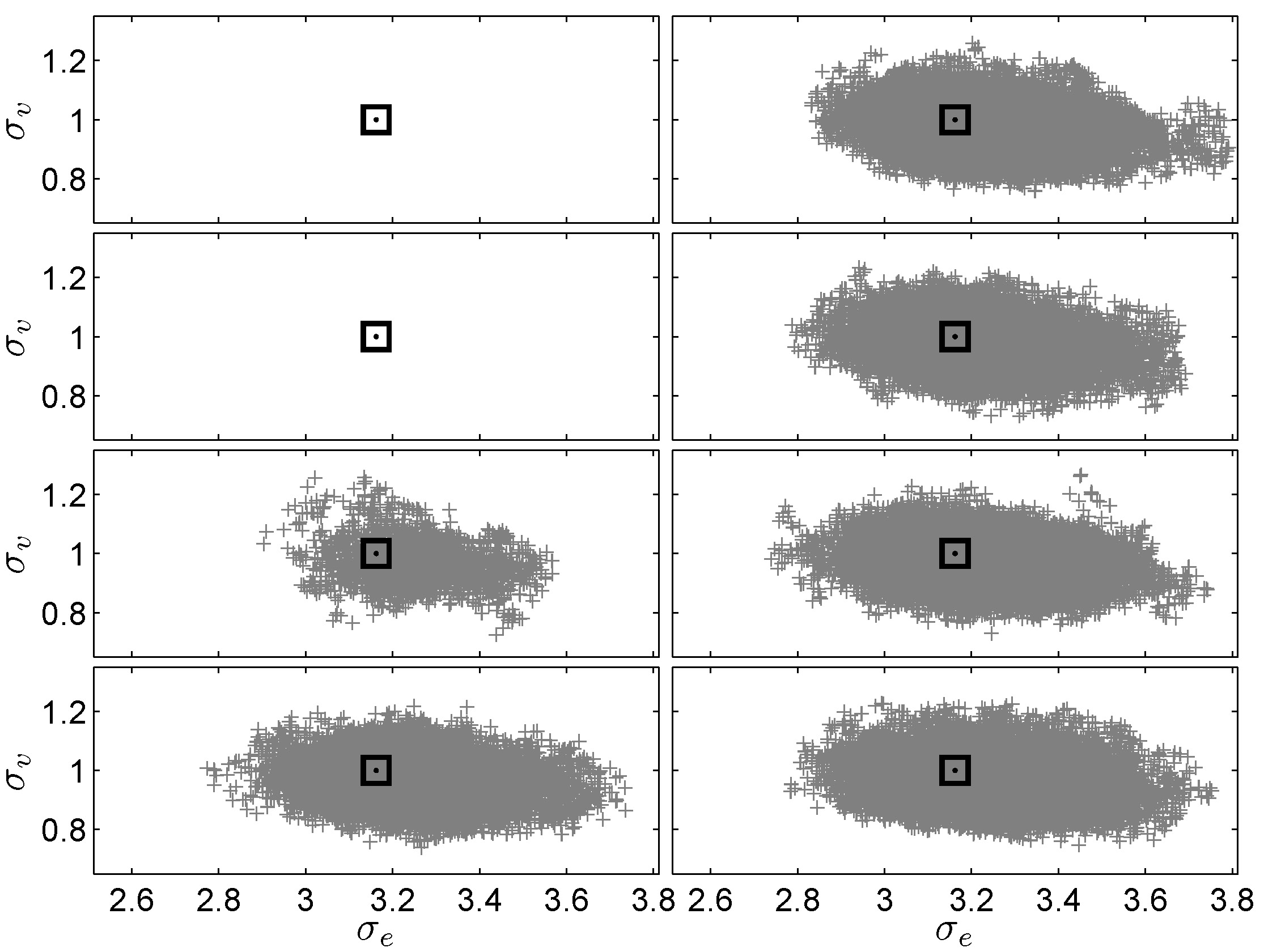}
    \caption{Scatter plots for $\{\sigma_v, \sigma_e\}$ for PMMH (left column) and \mwpg (right column). From top to bottom;
      $\Np = 5$, $\Np = 100$, $\Np = 1000$, $\Np = 5000$. 
      PMMH using $\Np = 5$ and $\Np = 100$ does not find the correct region of high posterior probability, and the samples lie outside
      the axes.}
    \label{fig:mwpg_scatter}
  \end{figure}

  \begin{figure}[H]
    \centering
    \includegraphics[height = 4.3cm]{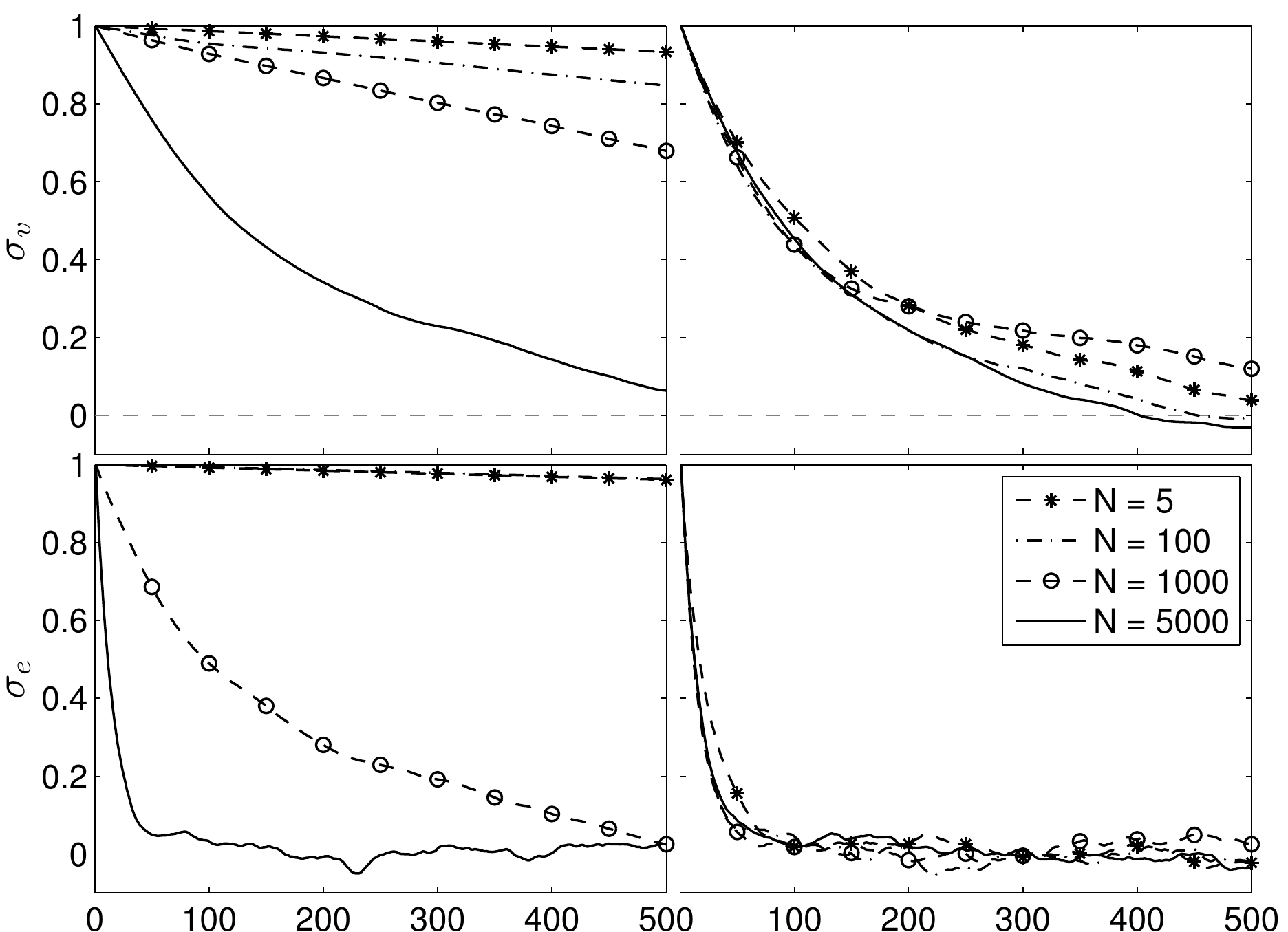}
    \caption{Autocorrelation plots for $\sigma_v$ (top row) and $\sigma_e$ (bottom row) for PMMH (left column) and \mwpg (right column). There is
      a sharp drop in autocorrelation for \mwpg for any number of particles.}
    \label{fig:mwpg_acf}
  \end{figure}
\end{multicols}

\begin{table}[ptb]
  \caption{Average acceptance probabilities}
  \centering
  \begin{tabular}{rcccccc}
    \toprule
    $\Np$ & $5$ & $100$ & $800$ & $2500$ & $7000$ & $10000$ \\
    \midrule
    PMMH & $1.4\cdot 10^{-4}$ & $8.6\cdot 10^{-3}$ & $0.082$ & $0.27$ & $0.45$ & $0.52$\\
    \mwpg & $0.62$ & $0.62$ & $0.61$ & $0.61$ & $0.61$ & $0.61$ \\
    \bottomrule
  \end{tabular}
  \label{tab:mwpg_accept}
\end{table}

\subsection{Numerical illustration -- estimating all parameters}\label{sec:mwpg_example2}%
So far, we have focused on estimating the noise variances in the model \eqref{eq:example_nl1d}. However,
Johannes, Polson and Yae \cite{JohannesPY:2010} suggested to identify all the parameters
$\parameter = \{\beta_1, \beta_2, \beta_3, \alpha, \sigma_v^2, \sigma_e^2\}$ in model \eqref{eq:example_nl1d},
as a more challenging test of PMCMC. They use the same parameter values as in \Sec{pgbsi_example},
but with $\sigma_v^2 = 1$ and $\sigma_e^2 = 10$, and apply an alternative, SMC based method with 300000 particles to estimate the parameters.
Here, we address the same problem with \mwpg using $\Np = 5$ particles and 60000 MCMC iterations.

We use vague normal distributed priors for $\beta_1$, $\beta_2$ and $\beta_3$, a uniform prior over $[1,\,3]$ for $\alpha$ and
inverse gamma priors with $a = b = 0.01$ for $\sigma_v^2$ and $\sigma_e^2$. To put even more pressure on the method,
Andrieu, Doucet and Holenstein \cite[p.~338]{AndrieuDH:2010} suggested to use a larger data set. Therefore, we use
$\T = 2000$ samples, generated from the model \eqref{eq:example_nl1d} and apply the \mwpg sampler. To propose values
for $\alpha$, we use a Gaussian random walk with standard deviation 0.02, constrained to the interval $[1,\,3]$.
For the remaining parameters, we make use of conjugacy of the priors and sample exactly from the conditional posterior
densities. \Fig{mwpg_histograms} show the estimated, marginal
posterior parameter densities for the six parameters, computed after a burnin of 10000 iterations. The method
appears to do a good job at finding the posterior density, even in this challenging scenario and using only 5 particles.

\begin{figure}[ptb]
  \centering
  \includegraphics[width = 0.7\columnwidth]{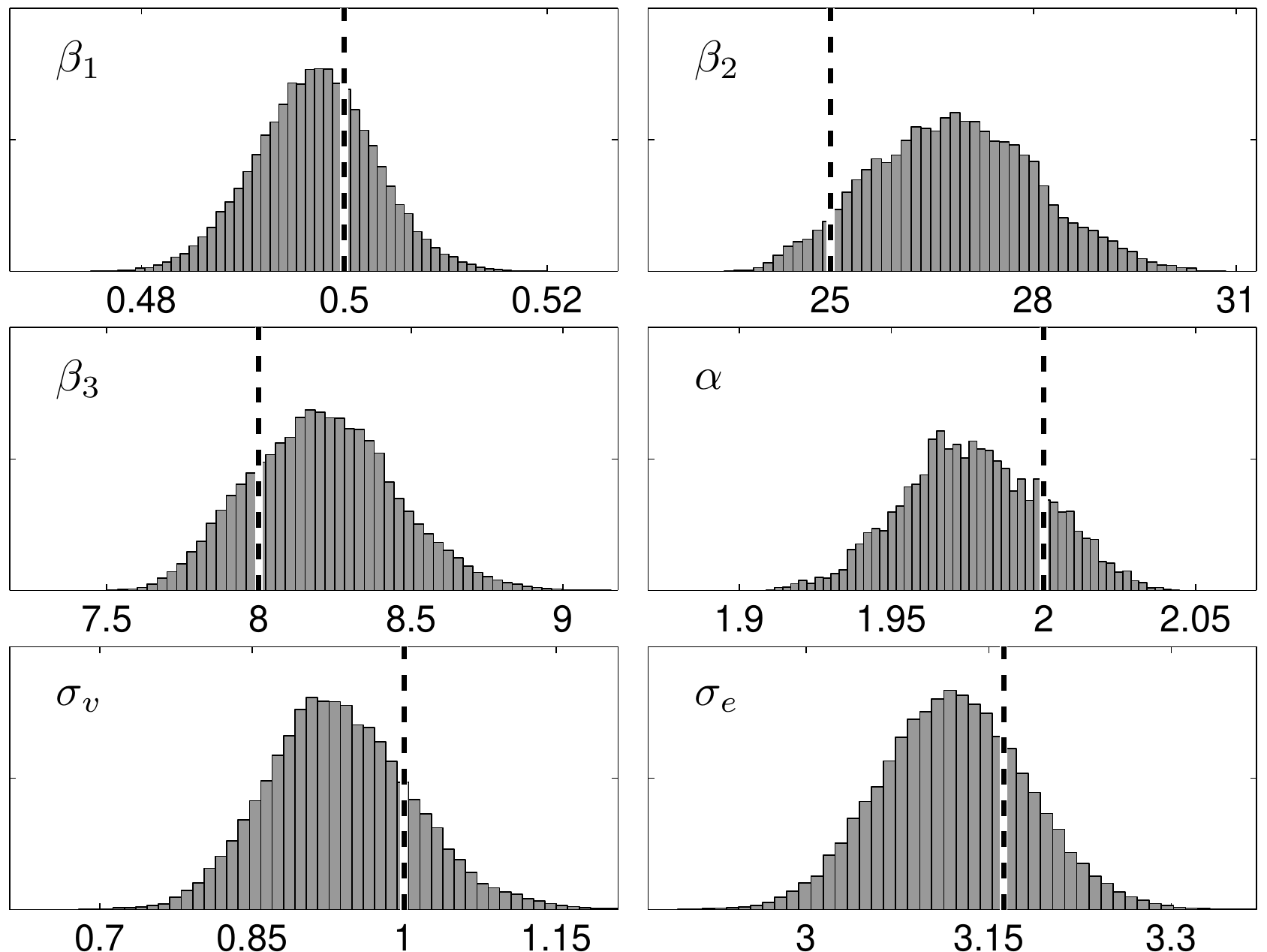}
  \caption{Posterior densities for the parameters of model \eqref{eq:example_nl1d}.
    The true values are marked by vertical dashed lines.}
  \label{fig:mwpg_histograms}
\end{figure}

\section{\pgbsi for joint smoothing}\label{sec:smoothing}%
\subsection{Collapsed \pgbsi}
Assume now that the parameter $\parameter$ is known and that our objective is to solve the smoothing problem for the model \eqref{eq:intro_ssm}.
That is, we seek the joint smoothing density $p_\parameter(x_{1:\T} \mid y_{1:\T})$ or some marginal thereof.
Several particle based methods have been presented in the literature, devoted to this problem, see \eg
 \cite{GodsillDW:2004,BriersDM:2010,DoucGMO:2011,FearnheadWT:2010}.
Many of these have quadratic complexity in the number of particles. Though, quite recently some advancements have been made towards linear
complexity, see \cite{DoucGMO:2011,FearnheadWT:2010}. In \cite{AndrieuDH:2010}, the PIMH sampler is suggested as a PMCMC method targeting
the joint smoothing density. The benefit of this is that the stationary distribution of the PIMH is the \emph{exact} smoothing distribution,
removing any bias that normally is present in more ``traditional'' particle smoothers. However, if we run the PIMH for $R$ iterations using $\Np$
particles, the computational complexity is $\Ordo(\Np R)$. In order to obtain a reasonable acceptance probability $\Np$ needs to be fairly high,
rendering the method computationally demanding. It is desirable to find a similar method that works even with a low number of particles.

The alternative that we propose in this paper is very straightforward. We simply run a \pgbsi according to \Alg{pg_bsi}, but since
$\parameter$ is fixed we skip step 2(a). This collapsed \pgbsi sampler will, similarly to PIMH, generate a Markov chain whose stationary distribution
is the exact smoothing distribution. Though, a major difference between the two methods is that the collapsed \pgbsi sampler will accept
all the generated particle trajectories, whereas the PIMH only will accept a portion of them in an MH fashion. The acceptance probability of
the PIMH will typically deteriorate as we decrease the number of particles $\Np$ or increase the length of the data record $\T$
(see \cite[Figure~3]{AndrieuDH:2010}).

The complexity of the collapsed \pgbsi sampler is still $\Ordo(\Np R)$, but since $\Np$ is allowed to be fairly small in this case, we believe that
this method could be a serious competitor to existing particle smoothers. Note that the method is exactly linear in the number of generated trajectories $R$.
Furthermore, the PMCMC based smoothers are trivially parallelisable, simply by running multiple parallel chains.

\subsection{Multiple trajectories}
One possible extension, to further improve the performance of the collapsed \pgbsi, is to use multiple trajectories as suggested
for the PIMH sampler in \cite{OlssonR:2011}.
Assume that we wish to estimate $\E[h(x_{1:\T}) \mid y_{1:\T}]$
for some function $h$ and that we run $R$ iterations of the \pgbsi sampler (possibly with some burnin).
The most straightforward estimator of $h$ is then
\begin{align*}
  \hat h = \frac{1}{R}\sum_{r=1}^R h(x_{1:\T}(r)).
\end{align*}
An alternative is to use the Rao-Blackwellised estimator
\begin{align*}
  \hat h_\text{RB} = \frac{1}{R}\sum_{r=1}^R \E \left[ h\left (x_{1:\T}(r) \right) \Mid \x_{1:\T}(r), \i_{2:\T}(r) \right].
\end{align*}
To compute the expectation in this expression, we would have sum over
all possible backward trajectories. Since there are $N^\T$ possible trajectories, this is in general not feasible,
unless $h$ is constrained to be a function of only some small subset of the $x$:s, \eg a single $x_t$ or a pair $\{x_t, x_{t+1}\}$.
However, a compromise would be to generate multiple backward trajectories in the backward simulator, providing an estimate of the
above expectation. Hence, at iteration $r$ of the \pgbsi sampler, instead of generating just a single backward trajectory
as in \Alg{pg_bsi}, we generate $M$ trajectories, $\tilde x_{1:\T}^m(r)$ for $m = \range{1}{M}$ and use the estimator
\begin{align*}
  \hat h_M = \frac{1}{RM} \sum_{r=1}^R \sum_{m=1}^M h(\tilde x_{1:\T}^m(r)).
\end{align*}
See \cite{OlssonR:2011} for a more in-depth discussion on using multiple trajectories in the context of PIMH.
Since the backward simulator generates conditionally independent and identically distributed (\iid) samples (given $\{\x_{1:\T}, \i_{2:\T}\}$),
we can set the state of the Markov chain to any of the generated backward trajectories, \eg $x_{1:\T}(r) = \tilde x_{1:\T}^1(r)$.

\subsection{Numerical illustration}\label{sec:smoothing_example}%
To illustrate the applicability of the collapsed \pgbsi for joint smoothing, we again consider
the same batch of $\T = 500$ samples from the model \eqref{eq:example_nl1d} as used in \Sec{pgbsi_example} and \Sec{mwpg_example}.
We apply the PIMH, the collapsed \pgbsi
and a state of the art particle smoother, known as fast forward filter/backward simulator (fast FFBSi) derived in \cite{DoucGMO:2011}.
The PMCMC samplers are run using $\Np = 20$ particles for $M = 1000$ iterations. The fast FFBSi is run using $\Np = 1000$ particles and
backward trajectories. To obtain something that we can visualise, let us consider the marginal smoothing density at time $t = 151$,
\ie $p_\parameter(x_{151} \mid y_{1:\T})$. \Fig{smoothing_cdfs} shows the corresponding cumulative distribution function\footnote{To obtain the ``true'' CDF, we run a fast FFBSi with 50000 particles and 10000 backward trajectories.} (CDF),
as well as the empirical CDFs from the three methods.

One question of interest is if it is possible to avoid burnin, by using a good initialisation of the MCMC sampler.
In the example above we did not use any burnin and the chain was initialised by
a run of an (unconditional) PF using $\Np = 20$ particles, followed by a backward simulation.
Another heuristic approach
is to estimate the joint smoothing distribution by running a ``standard'' particle smoother, \eg the fast FFBSi, with a moderate number of particles.
We can then initiate the collapsed \pgbsi sampler by sampling from the estimated joint smoothing distribution, which is hopefully close to the
stationary distribution. This heuristic can be of particular interest it we are able to run several parallel chains.
We can then apply a fast FFBSi to generate $M$ backward trajectories,
which are conditionally \iid samples from the empirical smoothing distribution, to initiate $M$ independent chains.

\begin{figure}[ptb]
  \centering
  \includegraphics[width = 0.7\columnwidth]{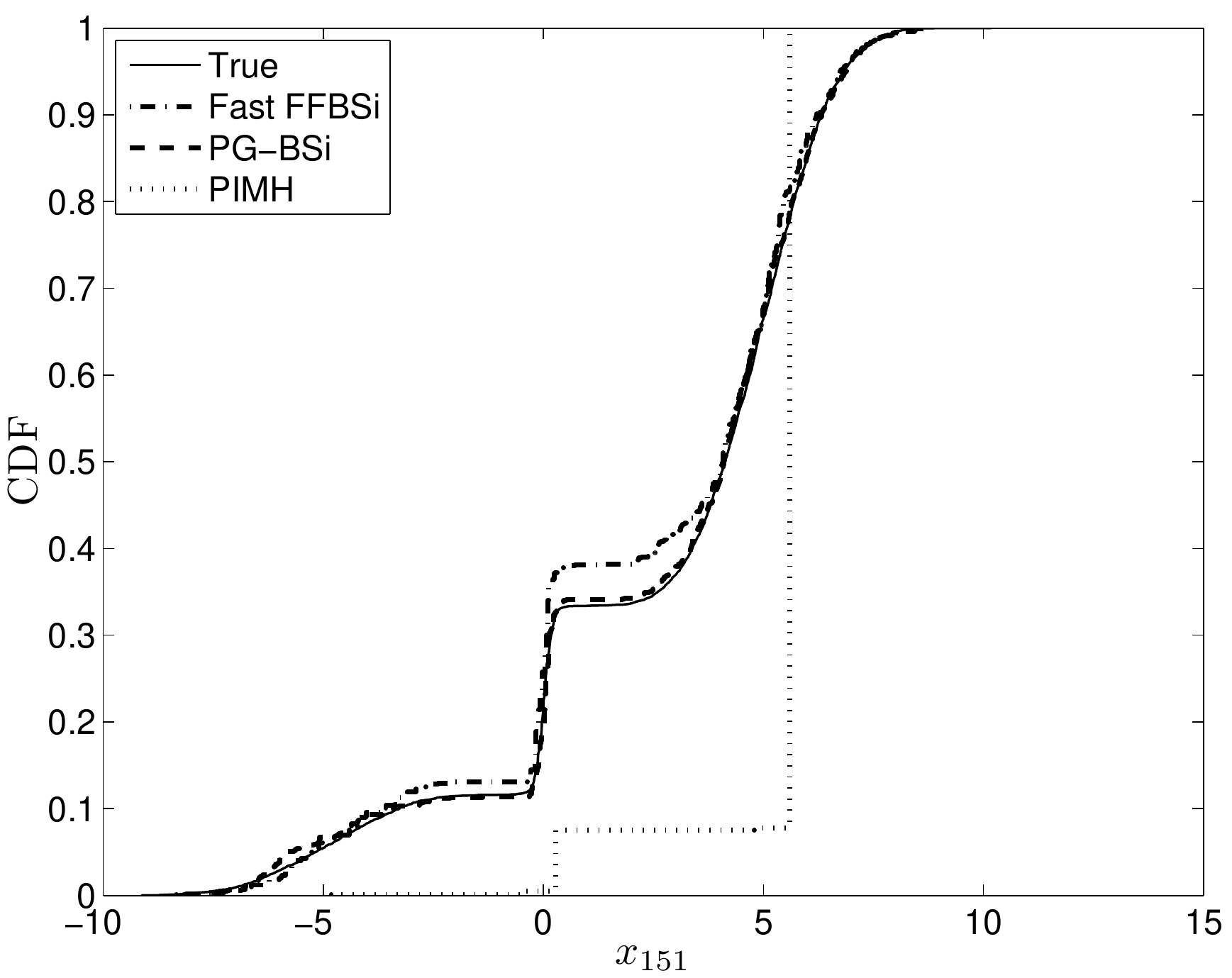}
  \caption{Empirical CDFs for the fast FFBSi, the collapsed \pgbsi and the PIMH. See text for details.}
  \label{fig:smoothing_cdfs}
\end{figure}

\section{Epidemiological model} 
As a final numerical illustration, we consider the identification of an epidemiological
model using the \mwpg sampler, based on search engine query observations.
Seasonal influenza epidemics each year cause millions of severe illnesses and hundreds of thousands of deaths
world-wide \cite{GinsbergMPBSB:2009}. Furthermore, new strains of influenza viruses can possibly
cause pandemics with very severe effects on the public health. The ability to accurately predict disease activity
can enable early response to such epidemics, which in turn can reduce their impact.

We consider a susceptible/infected/recovered (SIR) epidemiological model with environmental noise and seasonal
fluctuations \cite{KeelingR:2007,RasmussenRK:2011}, which is discretised according to the Euler-Maruyama method,
\begin{subequations}
  \label{eq:epi_SIR}
  \begin{align}
    S_{t+\sampletime} &= S_t + \mu P \sampletime - \mu S_t \sampletime - \left( 1 + v_t^S \right)\beta(t)\frac{S_t}{P}I_t \sampletime, \\
    I_{t+\sampletime} &= I_t - (\gamma + \mu)I_t \sampletime + \left( 1 + v_t^I \right)\beta(t)\frac{S_t}{P}I_t \sampletime, \\
    R_{t+\sampletime} &= R_t + \gamma I_t \sampletime - \mu R_t \sampletime.
  \end{align}
  Here, $S_t$, $I_t$ and $R_t$ represent the number of susceptible, infected and recovered individuals at time $t$ (months), respectively.
  The sampling time $\sampletime$ is chosen as 1/30 in our experiment, \ie we sample the model once a day.
  The state of the system is comprised of $x_t = \begin{pmatrix} S_t & I_t \end{pmatrix}^\+$.
  The total population size $P = 10^6$ and the host birth/death rate $\mu = 0.0012$ are both assumed known.
  The seasonally varying transmission rate is given by
  $\beta(t) = \bar \beta (1+ \alpha \sin (2\pi t/12))$ and the rate of recovery is given by $\gamma$.
  The process noise is assumed to be Gaussian according to
  $\begin{pmatrix} v_t^S & v_t^I \end{pmatrix}^\+ \sim \N(0, F^2 \eye{2}/\sqrt{\sampletime} )$.

  In \cite{RasmussenRK:2011}, the PMMH sampler is used to identify a similar SIR model, though with a different observation model
  than what we consider here (see below). A different Monte Carlo strategy, based on a particle filter with an augmented state space,
  for identification of an SIR model is proposed in \cite{SkvortsovR:2011}. Here, we
  use an observation model which is inspired by the Google Flu Trends project \cite{GinsbergMPBSB:2009}.
  The idea is to use the frequency of influenza related search engine queries to infer knowledge about the dynamics
  of the epidemic. As proposed by \cite{GinsbergMPBSB:2009}, we use a linear relationship between the observations
  and the log-odds of infected individuals, \ie
  \begin{align}
    \label{eq:epi_obs}
    y_t &= \rho \log \left( \frac{I_t}{P-I_t}\right) + e_t, & e_t &\sim \N(0,\tau^2).
  \end{align}
\end{subequations}
The parameters of the model are $\parameter = \{\gamma, \bar\beta, \alpha, F, \rho, \tau\}$,
with the true values given by 
$\gamma = 3$, $\bar\beta = 30.012$, $\alpha = 0.16$, $F = 0.03$ , $\rho = 1.1$ and $\tau = 0.224$.
We place a normal-inverse-gamma
prior on the pair $\{ \rho, \tau^2 \}$, \ie $p(\rho, \tau^2) = \N(\rho ; \mu_\rho, c_\rho \tau^2) \mathcal{IG}(\tau^2 ; a_\tau, b_\tau)$.
The hyperparameters are chosen as $\mu_\rho = 1$, $c_\rho = 0.5$ and $a_\tau = b_\tau = 0.01$. We place an inverse gamma prior on $F^2$ with hyperparameters
$a_F = b_F = 0.01$ and a normal prior on $\gamma$ with $\mu_\gamma = 3$ and $\sigma_\gamma^2 = 1$. The above mentioned priors are conjugate
to the model defined by \eqref{eq:epi_SIR}, which is exploited in the \mwpg sampler by drawing values from their conditional
posteriors. For the parameters $\bar\beta$ and $\alpha$ we use flat improper priors over the positive real line and sample new values according
to MH moves, as described in \Sec{mwpg}. Samples are proposed
independently for $\bar\beta$ and $\alpha$, according to Gaussian random walks with variances $5\cdot 10^{-3}$ and $5\cdot 10^{-4}$, respectively.

We generate 8 years of data with daily observations. The number of infected individuals $I_t$ over this time period is
shown in \Fig{epi_I}. The first half of the data batch is used for estimation of the model parameters. For this cause,
we employ the \mwpg sampler using $\Np = 20$ particles for 100000 MCMC iterations.
Histograms representing the estimated posterior parameter distributions are
shown in \Fig{epi_hist}. We notice some raggedness in some of the estimated marginals. This implies that
they have not fully converged, suggesting a fairly slow mixing of the Markov chain.
Still, the estimated distributions seem to provide accurate information about the general shapes and locations
of the posterior parameters distributions, and the true parameters fall well within the credible regions of the posteriors.

Finally, the estimated model is used to make 1 month ahead predictions of the disease activity for the subsequent 4 years,
as shown in \Fig{epi_I}. The predictions are made by Monte Carlo sampling from the posterior parameter distribution, followed by
a run of a particle filter on the validation data, using 100 particles. As can be seen, we obtain an accurate prediction of the
disease activity, which falls within the estimated 95 \% credibility intervals, one month in advance. 

\begin{figure}[ptb]
  \centering
  \includegraphics[width = 0.7\columnwidth]{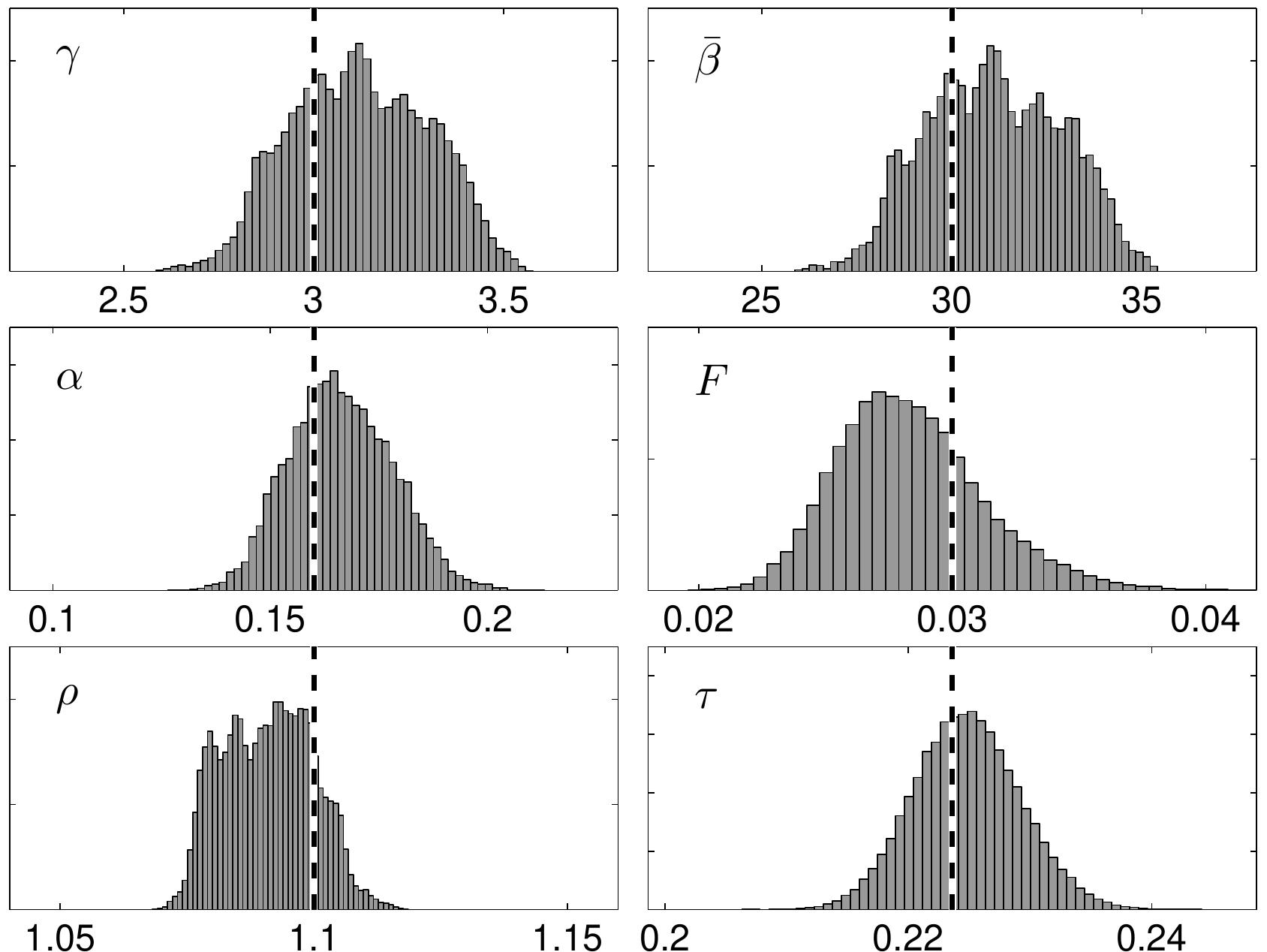}
  \caption{Posterior densities for the parameters of model \eqref{eq:epi_SIR}.
    The true values are marked by vertical dashed lines.}
  \label{fig:epi_hist}
\end{figure}

\begin{figure*}[ptb]
  \centering
  \includegraphics[width = 1\linewidth]{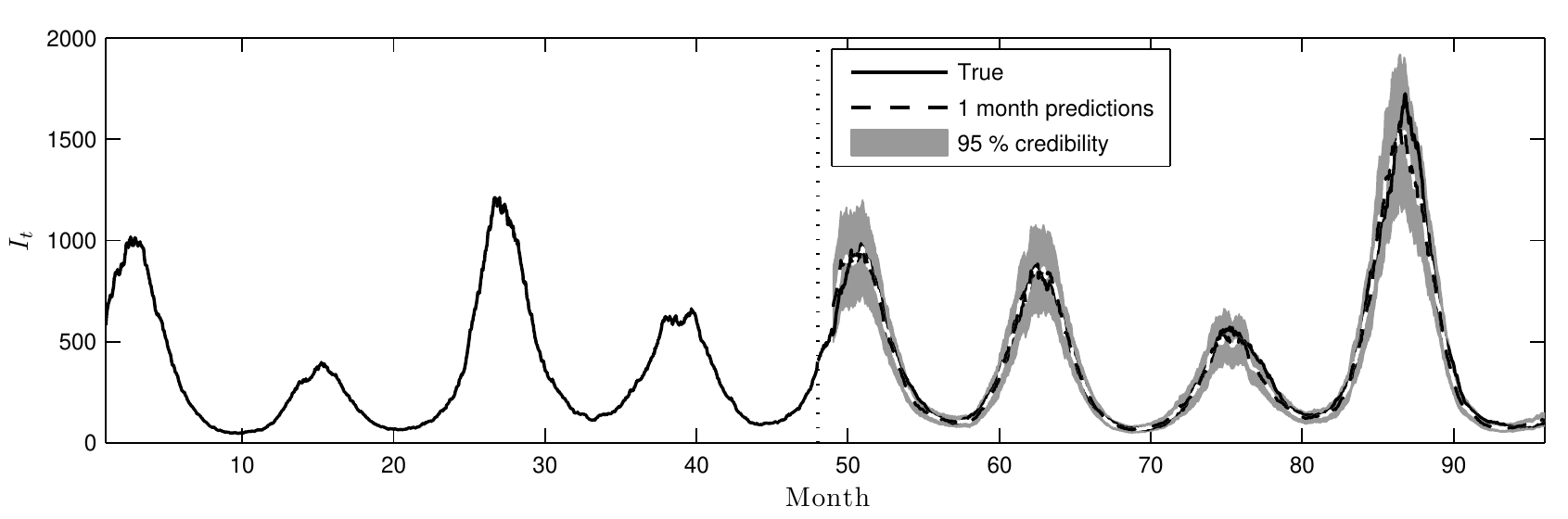}
  \caption{Disease activity (number of infected individualt $I_t$) over an 8 year period. The first 4 years are used as estimation data, to find the unknown parameters of the model.  For the consecutive 4 years, one month ahead predictions are computed using the estimated model.}
  \label{fig:epi_I}
\end{figure*}

\section{Conclusions and future work}\label{sec:conclusion}%
In the original PMCMC paper, Andrieu, Doucet and Holenstein writes \textit{``PMCMC sampling is much more robust and less likely to suffer from this
depletion problem} [referring to the PF]. \textit{This stems from the fact that PMCMC methods do not require SMC algorithms to provide a reliable
approximation of $p_\parameter(x_{1:\T} \mid y_{1:\T})$, but only to return a single sample approximately distributed according to
$p_\parameter(x_{1:\T} \mid y_{1:\T})$''}.

This observation is the very essence of the PMCMC methods and the reason for why they have shown to be such powerful tools
for joint Bayesian state and parameter inference. In this paper we have argued that the use of backward simulation in PMCMC
make this claim even stronger. Three methods have been presented, each similar to one
of the methods derived in \cite{AndrieuDH:2010}. These are; \pgbsi (similar to \pg), \mwpg (similar to PMMH) and collapsed \pgbsi (similar to PIMH).

We have shown experimentally that all of these methods work properly, even when we use very few particles in the underlying PFs.
This robustness to a decreased number of particles stems from a simple modification of the \pg sampler, namely to add a backward simulation
sweep to reduce the correlation between the particle trajectories sampled at any two consecutive iterations.
This modification does not only enable the application of PMCMC methods with very few particles,
but will also allow us to treat larger data sets.

It would be interesting to analyse how one should choose the number of particles. The mixing of the MCMC kernels will increase as we
increase $\Np$, but the improvement seems to saturate rather quickly (see \eg \Fig{mwpg_acf}). Based on this, we believe that
it in general is better to use fewer particles and more MCMC iterations. However, deducing to what extent this statement is problem
dependent is a topic for future work.

Another interesting avenue for future work is regarding theoretical aspects of the mixing of the \pgbsi sampler. That is,
whether or not it is possible to show that \pgbsi always will mix faster than \pg for the same number of particles.
We believe that this is the case, and some insight into the matter can be gained by considering the maximal correlations
for the two methods (see \eg \cite{LiuWK:1994} for a general discussion).

When it comes to a comparison between the \mwpg
and the PMMH samplers, it is not as clear to deduce which method that is preferable over the other. In fact, we believe
that this is very much problem dependent, and that the two methods are complementary.
If there is a strong dependence between the states and the parameters,
the \mwpg sampler should suffer from poor mixing, since it samples these groups of variables conditioned on each other.
Basically, this is the case if there is little noise in the system, or more severely, if the model is degenerate (\eg if noise only enters on some of the states).
In the latter case, the method might not even converge, since it is not clear that \Assumption{ergodicity} is fulfilled.
In such cases, the PMMH sampler might be a better choice. On the other hand, from the numerical examples that we have considered,
it is clear that the \mwpg sampler can substantially outperform the PMMH sampler in other cases. It would be of great interest
to conduct a more in-depth comparison between these two methods, to provide some general guidelines for when to use which method.

\appendix
\section{Proofs}\label{app:proofs}

\subsection{Proof of \Prop{pgbsi_capf}}
\Prop{pgbsi_capf} follows directly from \Prop{pgbsi_marginal} and the fact that
\begin{align*}
  \phi(\x_{1:\T}^{-j_{1:\T}}, \i_{2:\T} \mid \parameter, x_{1:\T}^{j_{1:\T}}, j_{1:\T})
  = \frac{\phi(\parameter, \x_{1:\T}, \i_{2:\T}, j_{1:\T})}{ \phi(\parameter, x_{1:\T}^{j_{1:\T}}, j_{1:\T}) }.
\end{align*}

\subsection{Proof of \Prop{pgbsi_bsi}}
We start by noting that, from the definition of the weight function \eqref{eq:pf_W}, we have for any $k \in \crange{1}{\Np}$,
\begin{align}
  \label{eq:proof_W}
  W_{t}^\parameter(x_t^{j_t}, x_{t-1}^{j_{t-1}}) = W_{t}^\parameter(x_t^{j_t}, x_{t-1}^{k})
  \times \frac{f_\parameter(x_{t}^{j_t} \mid x_{t-1}^{j_{t-1}}) \nu_{t-1}^\parameter(x_{t-1}^k, y_t) R_t^\parameter(x_t^{j_t} \mid x_{t-1}^k) }
       {f_\parameter(x_{t}^{j_t} \mid x_{t-1}^{k}) \nu_{t-1}^\parameter(x_{t-1}^{j_{t-1}}, y_t) R_t^\parameter(x_t^{j_t} \mid x_{t-1}^{j_{t-1}}) }.
\end{align}
Furthermore, it holds that
\begin{align*}
  &p_\parameter(x_{1:t}^{j_{1:t}} \mid y_{1:t}) \propto g_\parameter(y_t \mid x_t^{j_t}) f_\parameter(x_t^{j_t} \mid x_{t-1}^{j_{t-1}})
  p_\parameter(x_{1:t-1}^{j_{1:t-1}} \mid y_{1:t-1}),
\end{align*}
where, using \eqref{eq:proof_W} in the second equality,
\begin{align*}
  g_\parameter(y_t \mid x_t^{j_t}) f_\parameter(x_t^{j_t} \mid {}&x_{t-1}^{j_{t-1}}) 
  = W_{t}^\parameter(x_t^{j_t}, x_{t-1}^{j_{t-1}}) \nu_{t-1}^\parameter(x_{t-1}^{j_{t-1}}, y_t) R_t^\parameter(x_t^{j_t} \mid x_{t-1}^{j_{t-1}}) \\
  &=  W_{t}^\parameter(x_t^{j_t}, x_{t-1}^{i_t^{j_t}})
  \frac{f_\parameter(x_{t}^{j_t} \mid x_{t-1}^{j_{t-1}}) \nu_{t-1}^\parameter(x_{t-1}^{i_t^{j_t}}, y_t) R_t^\parameter(x_t^{j_t} \mid x_{t-1}^{i_t^{j_t}}) }
  {f_\parameter(x_{t}^{j_t} \mid x_{t-1}^{i_t^{j_t}}) } \\
  &\propto w_t^{j_t} \frac{ f_\parameter(x_{t}^{j_t} \mid x_{t-1}^{j_{t-1}}) }{ w_{t-1}^{i_t^{j_t}} f_\parameter(x_{t}^{j_t} \mid x_{t-1}^{i_t^{j_t}}) }
  M_t^\parameter(i_t^{j_t}, x_t^{j_t}).
\end{align*}
Using the above results, we can write
\begin{align*}
  &\phi(j_{1:\T} \mid \parameter, \x_{1:\T}, \i_{2:\T}) \\
  &\propto p_\parameter(x_{1:\T}^{j_{1:\T}} \mid y_{1:\T})
  \prod_{\exind{m}{j_1}}^\Np R_1^\parameter(x_1^m)
  \times \prod_{t = 2}^\T  \left( \prod_{\exind{m}{j_t}}^\Np M_t^\parameter(i_t^m, x_t^m) \right)
  \frac{w_{t-1}^{i_t^{j_t}} f_\parameter(x_t^{j_t} \mid x_{t-1}^{i_t^{j_t}})}{ \normsum{l} w_{t-1}^l f_\parameter(x_t^{j_t} \mid x_{t-1}^l)} \\
  &\propto p_\parameter(x_{1:\T-1}^{j_{1:\T-1}} \mid y_{1:\T-1})
  \prod_{\exind{m}{j_1}}^\Np R_1^\parameter(x_1^m)
  \times \prod_{t = 2}^{\T-1}  \left( \prod_{\exind{m}{j_t}}^\Np M_t^\parameter(i_t^m, x_t^m) \right)
  \frac{w_{t-1}^{i_t^{j_t}} f_\parameter(x_t^{j_t} \mid x_{t-1}^{i_t^{j_t}})}{ \normsum{l} w_{t-1}^l f_\parameter(x_t^{j_t} \mid x_{t-1}^l)}\\
  &\hspace{2cm}\times w_\T^{j_\T} \frac{f_\parameter(x_{\T}^{j_\T} \mid x_{\T-1}^{j_{\T-1}})} {\normsum{l}  w_{\T-1}^l f_\parameter(x_\T^{j_\T} \mid x_{\T-1}^l)}.
\end{align*}
By repeating the above procedure we get
\begin{align*}
   \phi(j_{1:\T} \mid \parameter, \x_{1:\T}, \i_{2:\T})
   \propto w_\T^{j_\T} \prod_{t = 2}^\T \frac{w_{t-1}^{j_{t-1}} f_\parameter(x_{t}^{j_t} \mid x_{t-1}^{j_{t-1}})}
   {\normsum{l}  w_{t-1}^l f_\parameter(x_t^{j_t} \mid x_{t-1}^l)},
\end{align*}
which completes the proof.

\subsection{Proof of \Thm{pgbsi}}
We have already assessed that $\phi$ is a stationary distribution of the \pgbsi sampler.
$\phi$-irreducibility and aperiodicity of the \pgbsi sampler under \Assumption{pf} and \Assumption{ergodicity} can be
assessed analogously as for the \pg sampler, see \cite[Theorem~5]{AndrieuDH:2010}. Convergence in total variation
then follows from \cite[Theorem~5]{AthreyaDS:1996}.

\bibliographystyle{IEEEtran}
\bibliography{references}

\end{document}